\documentclass[journal]{IEEEtran}
\usepackage{cite}

\usepackage{hyperref}

\ifCLASSINFOpdf
   \usepackage[pdftex]{graphicx}
   \graphicspath{{./figures/}}
   \DeclareGraphicsExtensions{.pdf,.jpeg,.png}
\else
\fi
\usepackage{amsmath}

\usepackage{amsfonts}
\usepackage{amsthm}
\usepackage{amssymb}

\usepackage{braket}

\newtheorem{thm}{Theorem}
\newtheorem{lem}{Lemma}
\newtheorem{defi}{Definition}
\newtheorem{rem}{Remark}

\usepackage[nameinlink,capitalise]{cleveref}

\Crefname{lem}{Lemma}{Lemma}
\Crefname{thm}{Theorem}{Theorem}

\usepackage{array}

\ifCLASSOPTIONcompsoc
  \usepackage[caption=false,font=normalsize,labelfont=sf,textfont=sf]{subfig}
\else
  \usepackage[caption=false,font=footnotesize]{subfig}
\fi
\usepackage{url}
\usepackage{xcolor}

\begin{document}
\title{Single-Shot Decoding of Linear~Rate LDPC~Quantum~Codes with High~Performance}
\author{Nikolas~P.~Breuckmann
	and~Vivien~Londe%
	
\thanks{N. P. Breuckmann, University College London, \href{mailto:n.breuckmann@ucl.ac.uk}{n.breuckmann@ucl.ac.uk}}%
\thanks{V. Londe, Team SECRET, INRIA, \href{mailto:vivien.londe@inria.fr}{vivien.londe@inria.fr}}%
}

\maketitle

\begin{abstract}
We construct and analyze a family of low-density parity check (LDPC) quantum codes with a linear encoding rate, polynomial scaling distance and efficient decoding schemes.
The code family is based on tessellations of closed, four-dimensional, hyperbolic manifolds, as first suggested by Guth and Lubotzky.

The main contribution of this work is the construction of suitable manifolds via finite presentations of Coxeter groups, their linear representations over Galois fields and topological coverings.
We establish a lower bound on the encoding rate~k/n of~13/72 = 0.180... and we show that the bound is tight for the examples that we construct.

Numerical simulations give evidence that parallelizable decoding schemes of low computational complexity suffice to obtain high performance.
These decoding schemes can deal with syndrome noise, so that parity check measurements do not have to be repeated to decode. Our data is consistent with a threshold of around~4\% in the phenomenological noise model with syndrome noise in the single-shot regime.
\end{abstract}


\IEEEpeerreviewmaketitle

\section{Introduction}

\IEEEPARstart{Q}{uantum} systems are susceptible to noise, which provides a formidable challenge to designing functioning and scalable quantum computers.
Noise prevents us from building even more powerful computing devices known as random access machines.
These are computers operating on analog signals and it can be shown that they can solve \textsc{PSPACE}-complete problems in polynomial time~\cite{RAMpower}.
However, small errors can build up uncontrollably in any analog computer.
This makes it impossible to scale these types of devices when noise is present and control is imperfect.
Shor showed that quantum computers are fundamentally different from analog computers in this regard, by showing that quantum errors can be dealt with by encoding the state of the quantum computer into a quantum code~\cite{shor_code}.
The accumulation of small errors is controlled by periodically performing measurements on the redundant degrees of freedom of the quantum code, thereby discretizing the error, and using the outcome of the measurement to determine a recovery operation.

A framework for the construction of quantum codes is provided by algebraic topology:
any manifold supporting a tessellation can be turned into a quantum code via its homology.
Well-knonw examples are the toric code, which is derived from a square tessellation of a torus and the surface code, which corresponds to the square tessellation of a topological disk~\cite{BK_surface_code,FM_surface_code}.
Properties of the code such as number of physical qubits~$n$, number of encoded qubits~$k$ and the code distance~$d$ are determined by the geometrical and topological properties of the tessellated manifold.

In~\cite{fetaya_bound,delfosse_bound} it was shown that the parameters of homological codes derived from 2D manifolds (surfaces) necessarily obey the bound 
\begin{align}
	kd^2 \leq \text{const.} \times (\log k)^2 n.
\end{align}

In~\cite{zemor} the author asked whether it is generally true that parameters of homological codes will be constrained by the bound $kd^2\in n^{1+o(1)}$.
The work of Guth and Lubotzky~\cite{guth_lubotzky} answered this question in the negative, by showing that codes derived from tessellations of four-dimensional hyperbolic manifolds have a linear encoding rate $k \sim n$ and polynomially scaling distance $d\in \Theta(n^\epsilon)$.
Their work left open how to actually construct these codes.

In this paper we discuss several approaches to this problem and explicitely construct closed, hyperbolic 4-manifolds from which we derive quantum codes.
We show that the code family has an asymptotic encoding rate~$k/n$ lower bounded by~$13/72$.
For the construction we consider regular tessellations of hyperbolic space. We will focus on a particular tessellation by a four-dimensional regular polytope called the {\em 120-cell}.
This polytope owes its name to the fact that its three-dimensional boundary consists of 120 dodecahedra.
The advantage of considering regular tessellations is that they can be described by their groups of symmetry, called {\em Coxeter groups}.
The first construction is based on finite presentations, which has been previously used to construct 2D hyperbolic codes~\cite{hyperbolic2d}.
A disadvantage of this approach is that finding closed manifolds is computationally expensive.
This problem is overcome by considering faithful representations of the Coxeter groups as matrix groups with coefficients in the ring~$\mathbb{Z}[\phi]$, where~$\phi$ is the golden ratio.
We relate the process of compactifying the infinte hyperbolic space~$\mathbb{H}^4$ to an algebraic procedure in terms of the linear representation.
It turns out that under certain conditions the symmetry group of the compactified space has a simple and well-known structure, allowing us to derive a formula for the size of the quantum code.
In order to obtain more examples of smaller size we use finite coverings, allowing us to construct spaces with less symmetries compared to the group-based constructions.
Finally, we perform Monte Carlo simulations to determine the performance of these codes.
We consider a decoder based on cellular automata~\cite{DKLP} as well as a decoder based on a message-passing algorithm, called belief-propagation.
Both decoding procedures have the advantage that they can be implemented using very simple classical control and are highly parallelizable.
The simulation results suggest that even when measurements are subject to noise it is possible to decode without having to repeat the measurement (single-shot error correction).
Even more encouraging is that the performance is higher than currently favoured quantum error correcting schemes.
Our data is consistent with an asymptotic threshold of~$p=4\%$ in the phenomenological $X/Z$-flip noise model with syndrome noise $q=p$.
This performance including measurement errors is higher than for a family of LDPC codes with similar parameters called  hypergraph product codes when assuming perfect measurements~\cite{hp_numerical}.

\subsection{Previous work}

Quantum codes based on hyperbolic 4-manifolds were originally proposed in \cite{guth_lubotzky} where it was shown that they possess a linear encoding rate and polynomially growing distance.
In \cite{hastings_decoder} a local decoding scheme was proposed and it was shown that under this scheme logical errors are polynomially suppressed.
Single examples of 4D~hyperbolic codes were constructed in \cite{golden_codes} and~\cite{homological_codes_thesis}. 
Examples of hyperbolic 4-manifolds with small volume were constructed in~\cite{belolipetsky} and~\cite{conder2005}.

\subsection{Summary}

In \Cref{sec:definition} we review the homological construction of quantum codes and the results obtained in~\cite{guth_lubotzky}.
In \Cref{sec:construction} we introduce regular tessellations of four-dimensional, hyperbolic space and their associated groups of symmetries and we derive the lower bound on the encoded rate for homological codes derived from such tessellations.
We then discuss the construction of closed, four-dimensional hyperbolic manifolds supporting regular tessellations using finitely presented groups and linear representations.
The list of examples is extended by considering less symmetric manifolds which are obtained by finite coverings.
We conclude the section by discussing the constructed examples in more detail.
Finally, in \Cref{sec:performance} we introduce simple decoding schemes and perform numerical simulations to determine the performance of the constructed code family.

\section{Definition and Properties}\label{sec:definition}

\subsection{Quantum Codes from Tessellated Manifolds}\label{sec:homological_codes}

Throughout the paper we assume that the number of physical qubits is $n$ and that their states  form a Hilbert space $\mathcal{H} = (\mathbb{C}^2)^{\otimes n}$.
A {\em quantum code}~$\mathcal{C}$ is a subspace of~$\mathcal{H}$ of dimension $2^k$ that is interpreted as the Hilbert space of~$k$ logical qubits.
Due to interactions with the environment error operators are applied randomly on the physical state.
It is assumed that such error operators act locally, meaning that they only act non-trivially on a small number of physical qubits.

A convenient class of quantum codes are called {\em stabilizer codes} where the code space is the $+1$-eigenspace of all elements of a subgroup $S$ of the Pauli group $P = \langle X_i, Y_i, Z_i \mid i\in \{1,\dotsc , n\} \rangle$ .
If the stabilizer group $S$ can be generated by operators which act as either purely $X$ or $Z$ then we call it a {\em CSS stabilizer code}.
CSS codes are closely related to binary linear codes from classical coding theory.
Given two binary linear codes of size $n$ with parity check matrices~$H_X$ and~$H_Z$ we can define a CSS stabilizer code simply by taking each row~$r$ of~$H_X$~($H_Z$) and defining an operator which acts as~$X$~($Z$) on qubit~$i$ if $r_i = 1$ and as the identity~$I$ otherwise.
Note that for the $+1$-eigenspace of $S$ to be non-trivial it is necessary that all of its generators commute.
This is achieved by demanding that
\begin{align}\label{eqn:CSS_condition}
	H_X \cdot H_Z^T = 0 .
\end{align}
Random constructions, which are commonly used in the classical setting, will generally not satisfy this constraint.
One way to find suitable parity check matrices~$H_X$ and~$H_Z$ is by considering homology over~$\mathbb{F}_2$, the field with two elements:
Given a closed manifold~$M$ of dimension~$D$ tessellated by polytopes, let $C_0$ be the $\mathbb{F}_2$-vector space which is formally generated by all vertices of the tessellation.
Similarly, we define $C_i$ as the vector space formally generated by all $i$-dimensional constituents of the tessellation (edges, faces, 3-cells,...).
We can now define boundary operators $\partial_i : C_i \rightarrow C_{i-1}$.
As each~$C_i$ comes with a distinguished basis we will always consider~$\partial_i$ as an $\mathbb{F}_2$-matrix with entries~$(\partial_i)_{m,n}$ equal to~$1$ if and only if the $i-1$-dimensional cell with label~$n$ is attached to the $i$-dimensional cell with label~$m$.
The elements of~$C_i$ can be identified with subsets of $i$-cells.
Applying $\partial_i$ to such an element will map it onto a subset of $i-1$-cells.
As contributions from neighboring $i$-cells will cancel modulo~2, we obtain that the result is the boundary of the initial subset.
An important observation is the fact that boundaries do not have boundaries themselves, which is equivalent to $\partial_{i} \circ \partial_{i+1} = 0$ for all $i=1,\dotsc , D-1$.

To define a CSS code we can simply define $H_X = \partial_{i}$ and~$H_Z = \partial_{i+1}^T$.
By doing so we have essentially identified $i$-cells with qubits, $i-1$-cells with $X$-checks and $i+1$-cells with $Z$-checks.

An alternative view on this construction is given by considering the tessellation as a partially-ordered set (poset).
The elements of the poset are all cells of the tessellation, where cells~$x$ and~$y$ fulfill the relation $x\prec y$ if and only if~$x$ is a subcell of~$y$ of one dimension lower.
The poset can be visualized as a diagram, as illustrated in \Cref{fig:poset}, where cells are nodes with two nodes $x$ and $y$ connected by an edge if and only if $x\prec y$.
As only cells with dimension differing by~1 are related the poset diagram forms a $D+1$-partite graph, where each partition is given by cells of a fixed dimension.
Picking any three consecutive layers we obtain what is called the Tanner graph of a CSS code: the middle layer forming the set of qubits and the outer two layers forming $X$-checks and $Z$-checks, respectively.
We note that the dual tessellation has the same poset diagram with the levels in reverse order.
If $i$ is chosen to be the middle dimension then $X$ and $Z$ are related by duality.

\begin{figure}
	\centering
	\includegraphics[width=0.7\columnwidth]{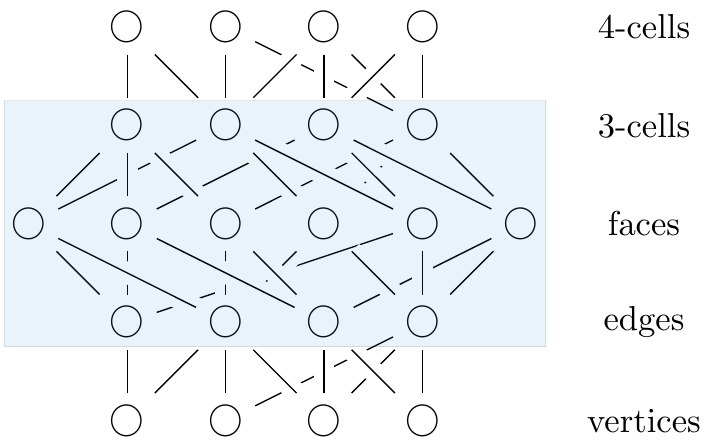}
	\caption{Poset diagram of a tessellation. The elements are cells and two related elements are connected by an edge. By definition of the relation (see main text) only cells with dimension differing by 1 are related. We can define a quantum code by picking three consecutive layers and define~$H_X$~($H_Z$) as the adjacency matrix between nodes in layers $i$ and $i-1$ ($i+1$). The box shows the case $i=2$. Note that any pair of an $i-1$-cell and an $i+1$-cell have an even number of $i$-cells that they are connected to in common, so that \Cref{eqn:CSS_condition} is satisfied. The subgraph in the box is the Tanner graph of the CSS code.}\label{fig:poset}
\end{figure}

The logical operators of a quantum code are characterized as those operators which commute with all checks while not being generated by them.
In particular, the logical $Z$-operators correspond to closed $i$-dimensional submanifolds which are not the boundary of an $i+1$-dimensional volume as they correspond to elements in~$\mathbb{F}_2^n$ which are in the kernel of the boundary operator, but not in its image.
They therefore correspond to elements of the {\em homology groups}~$H_i = \text{ker } \partial_{i} / \text{im } \partial_{i+1}$, where $i$-cells correspond to qubits.
Assuming that~$i$ is the middle dimension, the logical $X$-operators similarly correspond to closed $i$-dimensional submanifolds which are not the boundary of an $i+1$-dimensional volume in the dual tessellation.

A familiar example of this construction is the toric code which is obtained by a torus with a square tessellation.
The qubits are identified with edges ($i=1$) so that faces give $Z$-checks and vertices give $X$-checks.
The two non-contractible loops of the primal (dual) tessellation are identified with the logical $Z$ ($X$) operators.

It is common to be imprecise with the word {\em code}.
It can refer to a single instance, but also to a whole family of codes.
For our purposes here, a code family will be obtained from a sequence of manifolds with increasing volume which all come from the same tessellation, so that they all share the same local structure.

\subsection{Single-Shot Decoding}

Single-shot decoding was first discussed in~\cite{bombin_singleshot} in the context of the 3D gauge color code, although the results immediately apply to 4D homological codes as well.
The  main idea is that the syndrome, which is extracted by the measurement, contains redundancies.
This makes it possible to infer a recovery operation in the presence of syndrome noise, either by performing classical decoding on the syndrome first and then feed the fixed syndrome into the quantum code decoder.
Alternatively, it is known that cellular automata are robust against noise in the classical setting~\cite{grinstein} and there is numerical evidence that cellular automaton decoders applied to higher-dimensional quantum codes

Clearly, the recovery operation will in general not correct back to a code state and leave a residual error.
It is shown in~\cite{bombin_singleshot} that there exists a threshold below which a recovery is still possible by employing percolation type arguments to control the spread of errors.

\subsection{4D Hyperbolic Codes}\label{sec:guth_lubotzky_codes}

What makes the homological construction of the previous section appealing is that the properties of the code are determined by the underlying tessellated manifold.
In particular, the number of logical qubits $k$ is determined by its topology and the distance~$d$ is bounded by the minimum volume of a non-contractible submanifold.

We will now review the results of \cite{guth_lubotzky} on the encoding rate and distance of quantum codes derived from families of 4-dimensional hyperbolic manifolds.

\subsubsection{Encoding rate}

We will first discuss the number of logical operators $k$.
As mentioned in the introduction, hyperbolic manifolds give rise to quantum codes which have a linear rate $k\sim n$.
The linear rate of hyperbolic codes follows from the Chern--Gau\ss --Bonnet theorem, which relates the Euler characteristic
\begin{align}\label{eqn:euler_char}
	\chi(M) := \sum_{i=0}^{D} (-1)^i \dim H_i(M)
\end{align}
of a closed manifold $M$ of even dimension $D$ to the geometry of the manifold.
The exact statement is that
\begin{align}\label{eqn:chern_gauss_bonnet}
	\chi(M) = \frac{1}{(2\pi)^\frac{D}{2}} \int_{M} \text{Pf}(\Omega)
\end{align}
where $\text{Pf}(\Omega)$ is the Pfaffian of the curvature form of the Levi-Civita connection.
For a hyperbolic manifold the integral on the right-hand side is in fact equal to $(-1)^\frac{D}{2}\, 2\, \text{vol}(M) / \text{vol} (S^D)$~\cite{nakahara}.
Note that we always assume that~$M$ is connected, which implies that $\dim H_0 = \dim H_D = 1$.

For $D=2$ we can exactly solve for $\dim H_1$:
\begin{align}
	\dim H_1 = \frac{\text{area}(M)}{2\pi} + 2 
\end{align}
By tessellating $M$ with regular polygons we can define a quantum code with
\begin{align}
	k = \left(1-\frac{2}{r} -\frac{2}{s} \right)\, n + 2
\end{align}
where $r$ and $s$ are the weights of the $X$-checks and $Z$-checks~\cite{hyperbolic2d}.

For $D=4$ and $i=2$ we can not solve exactly for $k=\dim H_2$, since we do not know the dimensions of the odd homology groups.
However, as they both have a negative sign in the alternating sum we  obtain the lower bound
\begin{align}\label{eqn:volume_bound}
	\dim H_2 \geq 2 \frac{\text{vol}(M)}{\text{vol}(S^4)} - 2 .
\end{align}
Since $\text{vol}(S^4) = 8\pi^2 / 3$ this gives $\dim H_2 \geq 0.07\, \text{vol}(M) - 2$.
This establishes that a quantum code defined on a tessellation with uniform density of~$M$ will have linear rate $k\sim n$.
The value of the encoding rate $k/n$ will depend on the tessellation.
In \Cref{sec:encoding_rate} we derive a lower bound for the encoding rate of a quantum code based on a particular tessellation of 4D hyperbolic space.
This lower bound turns out to be tight for the examples we construct later (cf.~\Cref{sec:examples}).

\subsubsection{Distance}

For quantum codes derived from hyperbolic surfaces ($D=2$) one can establish upper and lower bounds on the distance which are logarithmic in the number of qubits~$n$.
For $D=4$ a lower bound on the distance follows from a result of systolic geometry by Anderson \cite{anderson}.
Let $R$ be the greatest length such that any ball of radius $R$ can be embedded anywhere in $M$.
This quantity is called the {\em injectivity radius of~$M$}.
Anderson's theorem states that any essential $i$-cycle $\gamma$ of~$M$ has its volume lower-bounded by the volume of a ball of radius~$R$ in $i$-dimensional hyperbolic space.
It is shown in~\cite{guth_lubotzky} that for a hyperbolic manifold $M$ we have $R \geq c\, \log \text{vol}(M)$ with a constant $c > 0$.
Combining this with Anderson's bound and the fact that the volume of a ball of radius~$r$ in~$\mathbb{H}^i$ grows like $\exp((i-1)\, r)$ we obtain that
\begin{align}
	\text{vol}(\gamma) \geq \text{vol}(B_R) = c'\, \exp\left( (i-1) R \right)
\end{align}
where $c'>0$ is a constant depending on $i$.
Hence we obtain for $i = 2$ that $\text{vol}(\gamma)$ is lower bounded by $c'\, \text{vol}(M)^{c}$.

\section{Construction and Examples}\label{sec:construction}

The discussion in \Cref{sec:definition} leaves open the question of how to obtain concrete examples of tessellations of closed hyperbolic 4-manifolds.
We will explain how we can describe tessellations using {\em Coxeter groups}, which are generated by reflections along all hyperplanes of symmetry of the tessellation.
We will review Coxeter groups in \Cref{sec:coxeter_groups}.
In particular, we will discuss how families of closed manifolds supporting a fixed tessellation are related to coverings of an infinite tessellation of $\mathbb{H}^4$.
We then give two separate constructions to obtain concrete examples of tessellated, closed, hyperbolic 4-manifolds as well as a method to obtain smaller, less symmetric manifolds from larger ones.

\subsection{Regular Tessellations}\label{sec:tessellations}

A {\em tessellation} is a gapless covering of a manifold by regular polytopes such that each adjacent pair of polytopes overlaps exactly on their facets.
We can decompose the regular polytopes into simplices by cutting them along their planes of symmetry.
We say that a tessellation is {\em regular} if the symmetry group of the tessellation operates transitively on these simplices.
This implies in particular that all polytopes are identical and that the same number of polytopes meet at every vertex, edge, face, etc.

Regular tessellations are classified by their Schl\"afli symbol $\{p,q,r,s,\dotsc\}$, which for a $D$-dimensional tessellation is a sequence of $D$~positive integers.
It encodes the incidence numbers of the cells:~q is the number of faces incident to a vertex in a 3-cell, r is the number of 3-cells incident to an edge in a 4-cell and s is the number of 4-cells incident to a face and so on.

Not every sequence of numbers corresponds to a valid tessellation of space due to geometric constraints.
For example, in 2D euclidean space the fundamental triangle of an~$\{r,s\}$ tessellation has internal angles $\pi/2$, $\pi/r$ and~$\pi/s$.
Since all internal angles have to add up to $\pi$ the only valid tessellations are the square tessellation $\{4,4\}$, the hexagonal tessellation~$\{6,3\}$ and the triangular tessellation~$\{3,6\}$.

The only possible regular tessellations of 4D hyperbolic space~$\mathbb{H}^4$ are:
\begin{enumerate}
	\item $\{5,3,3,5\}$ tessellation by 120-cells, self-dual
	\item $\{4,3,3,5\}$ tessellation by hypercubes
	\item $\{5,3,3,4\}$ tessellation by 120-cells, dual to 2
	\item $\{3,3,3,5\}$ tessellation by 4-simplices
	\item $\{5,3,3,3\}$ tessellation by 120-cells, dual to 4
\end{enumerate}
The 120-cell is a 4-dimensional regular polytope with Schl\"afli symbol~$\{5,3,3\}$ (see \Cref{fig:120cell}).
It has 120 dodecahedra~$\{5,3\}$ at its boundary.
Note that the dual tessellation has its Schl\"afli symbol reversed.
Compact 4-manifolds supporting the $\{5,3,3,3\}$ tessellation were constructed in~\cite{belolipetsky} and~\cite{conder2005}.
Quantum codes based on the $\{4,3,3,5\}$ tessellation were discussed in~\cite{golden_codes}. 

\begin{figure}
	\centering
	\subfloat[120-Cell]{\includegraphics[width=0.45\linewidth]{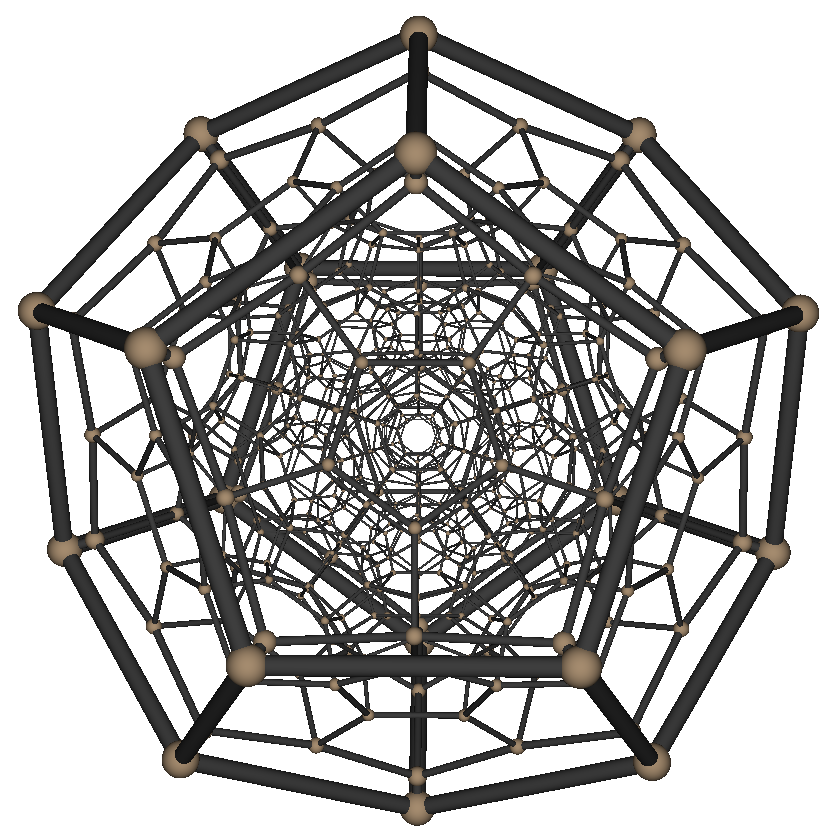}\label{fig:120cell}}
	\hfil
	\subfloat[Fundamental Simplex]{\includegraphics[width=0.45\linewidth]{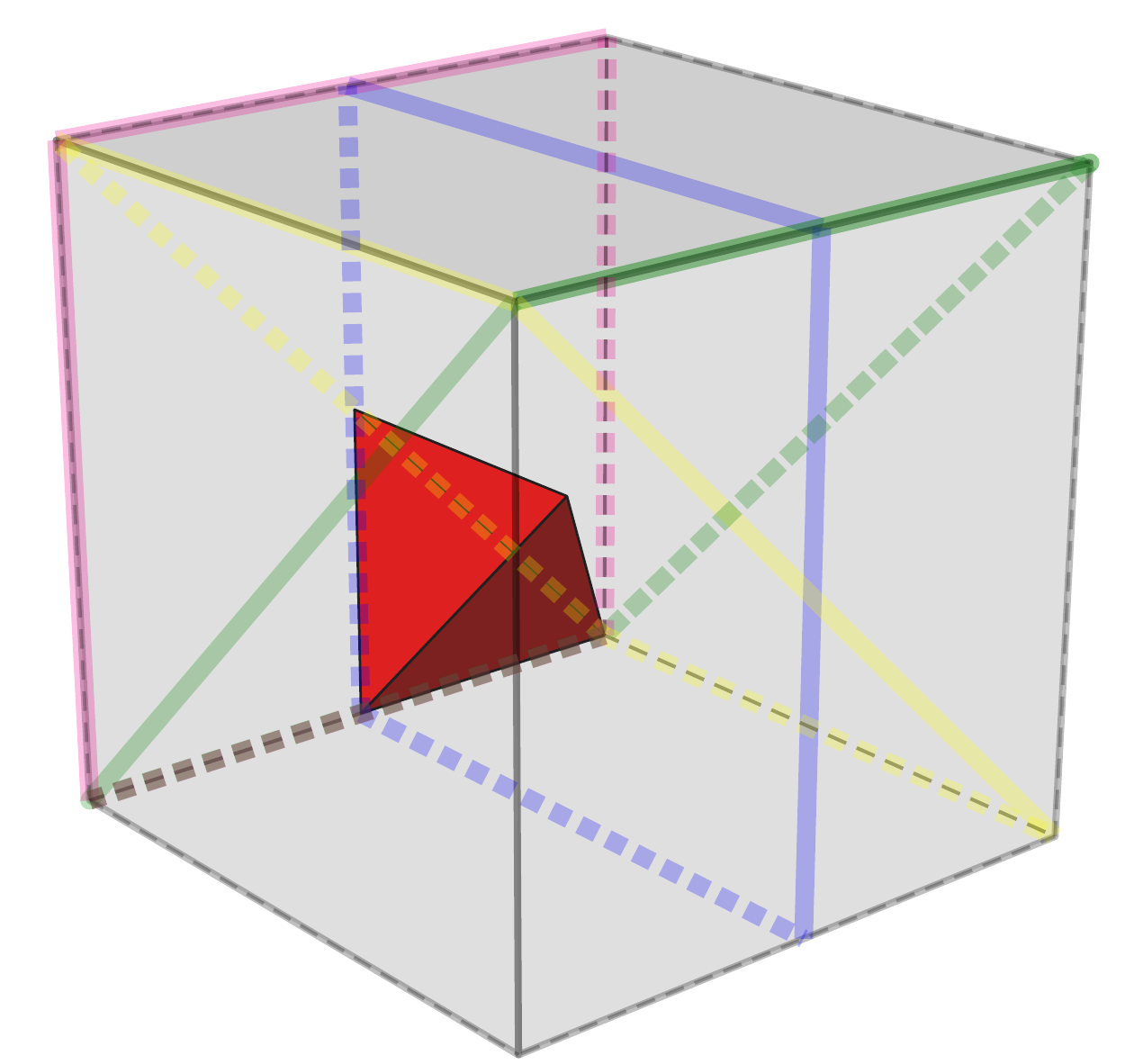}\label{fig:fundamental_domain}}
	\caption{(a) A 4D regular polytope called the 120-cell projected into 3D.
		 (b)~A single cube of a cubic tessellation $\{4,3,4\}$. The fundamental simplex is highlighted in red. It is bounded by the reflections~$a_0$, $a_1$, $a_2$ and~$a_3$, which are highlighted in blue, yellow, green and magenta, respectively. Each cube contributes 48 fundamental simplices.}
	\label{fig_sim}
\end{figure}

\subsection{Coxeter groups}\label{sec:coxeter_groups}
The group of symmetries of a regular tessellation is generated by reflections along hyperplanes of symmetry.
The hyperplanes of symmetry subdivide the tessellation into identical $D$-dimensional  simplices (see~\Cref{fig:fundamental_domain}).
The symmetry group acts freely and transitively on the simplices, meaning that no simplex is stabilized by the group action and every simplex can be mapped onto any other.
By fixing one arbitrary simplex and assigning it the identity element of the group, we have a one-to-one correspondence between the simplices and the group elements.

The Coxeter group is defined in terms of the generators and their relations.
As each generator $a_i$ corresponds to a reflection we have $a_i^2 = e$, where~$e$ is the neutral element of the group.
The relations between the generators are given by the Schl\"afli symbol
\begin{align}\label{eqn:coxeter_relations}
	(a_i\, a_j)^{r_{i,j}} = e
\end{align}
 where $r_{i,j}$ is the $j$th entry of the Schl\"afli symbol if $j = i+1$.
Note that the reflection relation gives $r_{i,i} = 1$.
All other pairs of generators (those with $|i-j|>1$) commute.
Since they are reflections this can be expressed as $(a_i\, a_j)^2=e$, i.e.~$r_{i,j} = 2$.

In the following chapters we will describe how we can use this description to obtain tessellations of compactifications of~$\mathbb{H}^4$.
Although the tools we present work for general tessellations, we will focus on the self-dual $\{5,3,3,5\}$ tessellation to construct quantum codes.

\subsection{Encoding Rate}\label{sec:encoding_rate}
Before discussing the constructions in the next few sections, we derive a lower bound on the encoding rate $k/n$ for codes derived from the $\{5,3,3,5\}$ tessellations introduced in \Cref{sec:tessellations}.

Instead of using the integral expression of the Chern-Gau\ss-Bonnet theorem of \Cref{eqn:chern_gauss_bonnet} we will instead consider the well-known combinatorial expression in terms of the number of cells in the tessellation.
In order to obtain this expression, we note that the number of $i$-cells is the same as the dimension of the vector space of $i$-chains~$C_i$. 
By the rank-nullity theorem and the definition of the homology groups $H_i = \text{ker } \partial_{i} / \text{im } \partial_{i+1}$ we have that 
\begin{align}
\begin{split}
\dim C_i &= \dim \ker \partial_i + \dim \text{im } \partial_i\\
&= \dim H_i + \dim \text{im } \partial_{i+1}  + \dim \text{im } \partial_{i}.
\end{split}
\end{align}
Putting this into the definition of the Euler charactistic (\Cref{eqn:euler_char}) we obtain
\begin{align}
\begin{split}
\chi = \sum_{i=0}^{D} (-1)^i \dim C_i = \sum_{i=0}^{D} (-1)^i \text{ \# i-cells}.
\end{split}
\end{align}
The number of cells can be expressed in terms of of the number of fundamental simplicies.
For the $\{5,3,3,5\}$ tessellation the number of fundamental simplices per vertex and 120-cell is both 14400, the number of fundamental simplicies per face is 100 and the number of fundamental simplices per edge and dodecahedron is both 240.
Let $S(M)$ be the total number of fundamental simplices of the tessellated manifold~$M$.
We obtain the following formula for the Euler characteristic:
\begin{align}
    \chi     &= \frac{13}{7200} S(M)
\end{align}
Together with \Cref{eqn:euler_char} we finally obtain the bound
\begin{align}\label{eqn:encoding_rate}
k \geq \frac{13}{72} n - 2
\end{align}
where the inequality is due to ignoring the negative contributions of $\dim H_1(M)$ and $\dim H_3(M)$. The constant term comes from $\dim H_0(M) = \dim H_4(M) = 1$.

We note that \Cref{eqn:volume_bound} and \Cref{eqn:encoding_rate} are consistent with one another, as the volume of a 4D hyperbolic manifold~$M$ is related its Euler characteristic via the equation $\text{vol}(M) = 4\pi^2 \chi(M)/3$ (see~\cite{gromov1982volume}).

\subsection{Construction based on FP-groups}\label{sec:constructionRWS}

We can use the identification between the fundamental simplices and the group elements to obtain tessellations of closed manifolds.
The idea is to consider finite quotients of the infinite group, which leave the local structure of the group invariant.
Geometrically, the procedure essentially consists of finding translations and identifying points which differ by these translations.
For example, on the 2D euclidean plane we can take an arbitrary translation and by identifying all points differing by this translation we obtain a cylinder of infinite length.
Taking a second translation, which is not co-linear with the first one, we obtain a torus.

This process is less straightforward in curved spaces where translations generally do not commute.
In \cite{hyperbolic2d} this has been done for 2D hyperbolic surfaces by enumerating normal subgroups and their quotients up to a certain size.
The Todd-Coxeter algorithm can be used to enumerate normal subgroups. Some faster adaptions using the Knuth–Bendix completion algorithm are also known (see Chapter~5.6 in~\cite{sims}).

By \Cref{eqn:coxeter_relations}, the group of the infinite $\{5,3,3,5\}$ tessellation is
\begin{align}
\begin{split}
	\langle a, b, c, d, e \mid& a^2, (ab)^5, (ac)^2, (ad)^2, (ae)^2,b^2, (bc)^3,\\ & (bd)^2, (be)^2, c^2, (cd)^3, (ce)^2, d^2, (de)^5, e^2 \rangle .
\end{split}
\end{align}
For readability we have written the generators as $a,\dotsc , e$ instead of $a_i$ for $i=0,\dotsc ,4$.
Trying to find normal subgroups of this group by exhaustive search yielded only two examples in a reasonable amount of time.
One example has 14,400 fundamental simplices and the other 72,000.
In the $\{5,3,3,5\}$ tessellation there are 100 simplices per face, so that we obtain quantum codes with~144 and 720 physicsl qubits, respectively (see~\Cref{tab:examples} and discussion in \Cref{sec:examples}).

We found larger examples by considering the following randomized procedure: we can take a random word in the generators of a specified length~$w$.  We then obtain a normal subgroup by taking its normal closure~$N$ of this group element and check if the resulting group is finite.
One additionally needs to check that $N$ operates fixed-point free which is the case if $w$ does not correspond to a reflection or a rotation~\cite{ratcliffe}.
This procedure gave two more examples with 18,432 and 19,584 physical qubits.

\subsection{Construction Based on Matrix Representations}\label{sec:constructionPSL}

The second construction is based on matrix representations of the symmetry groups.
The main idea is to obtain a faithful matrix representation of the infinite tessellation.
Let us assume that we are able to find a representation with a distinguished basis such that all of the generators and their inverses are mapped onto matrices which have integer entries.
Clearly, in this case all group elements are represented by integer matrices.
To obtain a finite group we could naively try to reduce the entries of all matrices modulo some positive integer~$p$.
This would ensure that we are left with a finite set of matrices.
There are some obvious problems with this approach:
It is generally not possible to have purely integer entries.
We will address these issues in what follows.

\subsubsection{Hyperboloid Model}
\begin{figure}
	\centering
	\includegraphics[width=0.9\columnwidth]{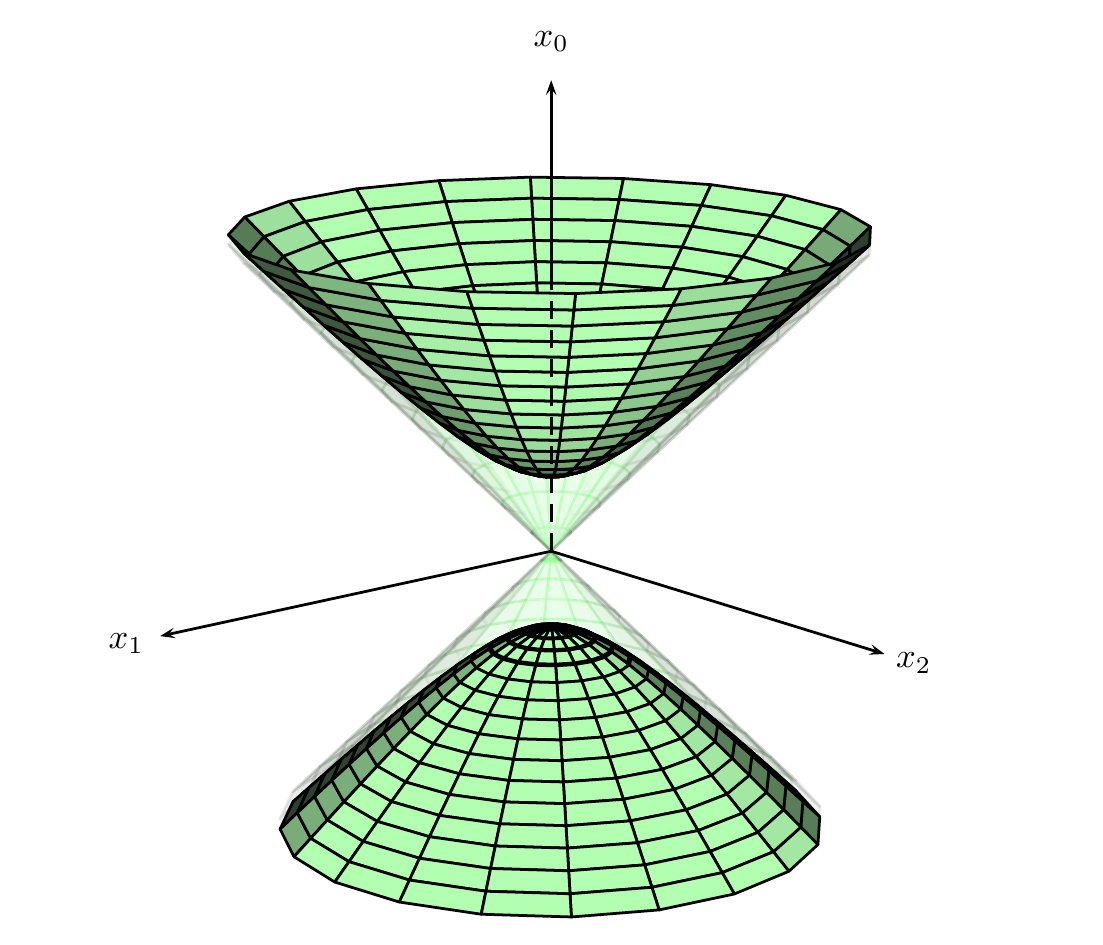}
	\caption{The hyperboloid model of hyperbolic space. The equation~$x\circ x = -1$ defines a hyperboloid consisting of two disconnected sheets. We identify the upper sheet ($x_0 >0$) with the hyperbolic plane.}\label{fig:hyperboloid_model}
\end{figure}

The matrix representation is obtained by the {\em hyperboloid model} of hyperbolic space:
In $D+1$-dimensional Minkowski space~$\mathbb{R}^{1,D}$ we can identify the $D$-dimensional hyperbolic plane with the set
\begin{align}
	\mathbb{H}^D = \Bigg\{ x\in \mathbb{R}^{1,D}\, \Bigg| \, x\circ x = - x_0^2 + \sum_{i=1}^{D} x_i^2 = -1,\, x_0 > 0 \Bigg\} 
\end{align}
by defining the distance between any two points $x,y\in \mathbb{H}^D$ as $\text{dist}(x,y) = \cosh^{-1}\left( - x\circ y \right)$ where $\circ$ denotes the Lorentzian inner product.
The group of invertible $(D+1)\times (D+1)$-matrices which leave the Lorentzian inner product invariant is called $\text{O}(1,D,\mathbb{R})$.
The isometry group of $\mathbb{H}^D$, i.e. the group of transformations which preserve the distance function $\text{dist}$ is isomorphic to the subgroup~$\text{O}^+(1,D,\mathbb{R})$ of index 2 which sends the upper sheet to itself and the lower sheet to itself (cf. \Cref{fig:hyperboloid_model}).
It is also known as the ``orthochronous Lorentz group'' in $D+1$-dimensions.

\subsubsection{Representaion of the infinite tessellation group}
We can construct the representation using the Gram-matrix~$g$ of the normal vectors of the hyperplanes of reflection.
Let~$e_1,\dotsc,e_5$ denote the standard basis vectors. The inner product between them is defined by
\begin{align}
	 g(e_i,e_j) = -2\, \cos\left(\frac{\pi}{r_{i,j}}\right)
\end{align}
where $r_{i,j}$ are the exponents in the relations which define the Coxeter group (see \Cref{sec:coxeter_groups}).
Since we consider regular tessellations $g$ will be tridiagonal with entries $\alpha_i = -2\, \cos(\pi/r_{i,i+1})$ on the first diagonals.
\begin{align}\label{eqn:inner_prod}
	g = \begin{bmatrix}
			2   & \alpha_0  & 0 & 0 & 0 \\
			\alpha_0   & 2 & \alpha_1 & 0 & 0 \\
			0   & \alpha_1 & 2 & \alpha_2 & 0 \\
			0   & 0 & \alpha_2 & 2 & \alpha_3 \\
			0   & 0 & 0 & \alpha_3 & 2
		\end{bmatrix}
\end{align}
If the Schl\"afli-symbol belongs to a hyperbolic tessellation then~$g$ has signature~$(-,+,+,+,+)$, i.e.~$g$ is equivalent, up to a change of basis, to the Lorentzian inner product~$\circ$.

The matrix representation $\rho : G \rightarrow \text{O}^+(1,D,\mathbb{R})$ can now be defined by their action on the basis vectors.
For each generator~$a_i$ of the Coxeter group its representation $\rho(a_i)$ is defined by its action on the standard basis:
\begin{align}\label{eqn:rep_def}
	\rho(a_i) \cdot e_j = e_j - 2\, \frac{g_{i,j}}{g_{i,i}}\, e_i = e_j - g_{i,j}\, e_i .
\end{align}

Let us verify that this is indeed a representation by explicitly checking the group relations.
To not clutter our notation we will write $r_i := \rho(a_i)$ (not to be confused with the matrix~$r$ from \Cref{sec:coxeter_groups}, which defines the relations).
Let us first check that we are indeed mapping onto reflections.
\begin{align}
\begin{split}
r_i^2(e_j) &= r_i\cdot (e_j - g_{i,j} e_i) \\
&= e_j - g_{i,j} e_i - g_{i,j} ( e_i - g_{i,i} e_i ) \\
&= e_j -2\, g_{i,j} e_i + 2\, g_{i,j} e_i\\
&= e_j
\end{split}
\end{align}
Hence,  $r_i$ is indeed a reflection.

We are left to show that the $r_i$ satisfy the ``off-diagonal relations'' of \Cref{eqn:coxeter_relations}.
Let $|i-j|>0$ and define~$v^{\perp_g}$ as the  space of all vectors orthogonal to $v$ with respect to~$g$.
Since $\dim e_i^{\perp_g} = \dim e_j^{\perp_g} = D$ and $e_i^{\perp_g} \neq e_j^{\perp_g}$ we have
\begin{align}
	\dim \left(e_i^{\perp_g} \cap e_j^{\perp_g}\right) = D-1. 
\end{align}
Since~$e_i$ and~$e_j$ do not belong to $e_i^{\perp_g} \cap e_j^{\perp_g}$ we can complete a basis of $e_i^{\perp_g} \cap e_j^{\perp_g}$ with~$e_i$ and~$e_j$ to form a basis~$\mathcal{B}$ of~$\mathbb{R}^{D+1}$. 
Let us express $r_i$ and $r_j$ in $\mathcal{B}$:
\begin{align}
\begin{split}
r_{i,\mathcal{B}} = \begin{bmatrix}
I_{D-1} & 0 & 0 \\
0 & -1 & g_{i,j} \\
0 & 0 & 1
\end{bmatrix}\\
r_{j,\mathcal{B}} = \begin{bmatrix}
I_{D-1} & 0 & 0 \\
0 & 1 & 0 \\
0 & g_{i,j} & -1
\end{bmatrix}
\end{split}
\end{align}
Their product is 
\begin{align}
r_{i,\mathcal{B}} \, r_{j,\mathcal{B}} = \begin{bmatrix}
I_{D-1} & 0 & 0 \\
0 & (g_{i,j})^2 - 1 & - g_{i,j} \\
0 & g_{i,j} & -1
\end{bmatrix} .
\end{align}

We can focus on the bottom-right two by two submatrix~$A$.
Its determinant is $\det(A) = 1$ and its trace is
\begin{align}
\begin{split}
\text{tr}(A) &= (g_{i,j})^2 - 2 \\
&= 4 \cos^2(\pi / r_{i,j}) -2 \\
&= 2 \cos(2 \pi / r_{i,j}) .
\end{split}
\end{align}
Therefore $A$ has two distinct eigenvalues $\lambda_+ = \exp(i\, 2 \pi / r_{i,j})$ and $\lambda_- = \exp(-i\, 2 \pi / r_{i,j})$ and hence satisfies $A^{r_{i,j}} = I_2$. 
We have thus shown that $(r_{i,\mathcal{B}} \, r_{j,\mathcal{B}})^{r_{i,j}} = I_n$.

We refer to Theorem 3A10 in~\cite{abstractregularpolytpes} for a proof that the representation~$\rho$ is faithful, i.e. the reflections $r_i$ do not satisfy other relations than the ones satisfied by the generators $a_i$ of the Coxeter group, so that $\text{im}(\rho) \simeq G$.
Note that~$\text{im}(\rho)$ is isomorphic a subgroup of~$\text{O}(1,D,\mathbb{R})$ as $g$ is equivalent to the Lorentzian inner product~$\circ$.

\subsubsection{Matrix entries}\label{sec:matrix_entries}
Our initial goal,  which we stated at the beginning of this section, was to obtain a matrix representation that allows us to take all elements modulo a large number.
Here we will see that this is generally not possible and we will show how to amend the idea to make it work.

What are the entries of the elements of~$\text{im}(\rho)$? -- Let us consider the self-dual tessellation wit Schl\"afli symbol~$\{5,3,3,5\}$.
From \Cref{eqn:rep_def} it is clear that all matrices have entries that are integer polynomials of the entries of~$g$.
The entries on the first diagonals of~$g$ (see \Cref{eqn:inner_prod}) are
\begin{align}
	\alpha_0 = \alpha_3 = -2\, \cos\left(\frac{\pi}{5}\right) = \frac{1+\sqrt{5}}{2} =: \phi
\end{align}
and
\begin{align}
	\alpha_1 = \alpha_2 = -2\, \cos\left(\frac{\pi}{3}\right) = -1 .
\end{align}
The matrix entries $\alpha_1$ and~$\alpha_2$ take integer values.
However,~$\alpha_0$ and~$\alpha_3$ are equal to the golden ratio~$\phi$, which is not an integer.
We can account for this by simply extending the ring of integers by $\phi$ and obtain $\mathbb{Z}[\phi]$.
Note that $\phi$ still fulfills the relation $\phi^2 -\phi -1 = 0$.
The associated polynomial $h = x^2-x-1$ is called the {\em minimal polynomial of~$\phi$}.
All generators are self-inverse and hence all matrices in $\text{im}(\rho)$ have entries in~$\mathbb{Z}[\phi]$.
To be able to use a computer algebra system we construct the ring~$\mathbb{Z}[\phi]$ from polynomials.
This can be achieved by considering all polynomials up to arbitrary multiples of $h$.
The set of all multiples of $h$ are called the {\em ideal generated by~$h$} and denoted
\begin{align}
	\langle h \rangle = \lbrace p\cdot h\mid p\in \mathbb{Z}[x] \rbrace .
\end{align}
Note that $\langle h \rangle$ is by definition closed under linear combinations.
The quotient ring~$\mathbb{Z}[x] / \langle h \rangle$ contains elements of the form~$p + \langle h \rangle$ with $p\in \mathbb{Z}[x]$.
In particular, if~$p$ is a multiple of~$h$ we have $p + \langle h \rangle = 0 + \langle h \rangle$.
This means that $x$ fulfills the same relations as~$\phi$ in~$\mathbb{Z}[x]$ and hence we have
\begin{align}
	\mathbb{Z}[\phi] \simeq \mathbb{Z}[x] / \langle h \rangle .
\end{align}
In the remainder of the paper we will abuse notation and directly identify $\mathbb{Z}[\phi]$ with $\mathbb{Z}[x] / \langle h \rangle$.

\subsubsection{Quotient}
In the previous paragraphs we have obtained a faithful matrix representation~$\rho$ of the symmetry group of the infinite $\{5,3,3,5\}$ tessellation of~$\mathbb{H}^4$.
We have seen that each element of~$\text{im}(\rho)$ has coefficients in~$\mathbb{Z}[\phi]$.

Our strategy to obtain symmetry groups of  closed hyperbolic four-manifolds is to factor out suitable ideals of~$\mathbb{Z}[\phi]$ to effectively obtain representations of $G$ over $\mathbb{F}_q^{5}$.
We need to show that factoring out ideals of the matrix entries does preserve the local structure, which means that the result should be the symmetry group of a closed, tessellated manifold that looks identical to the infinite tessellation in a large neighborhood.

Our goal is to obtain a family of quantum codes with growing distance.
To show that the distance increases it suffices to show that the tessellation on the closed manifold is indistinguishable from the infinite one in a large neighborhood, as no logical operator can have support inside this neighborhood.

Let $l$ be a positive integer.
We call a representation {\em $l$-locally faithful} if no non-identity element $g\in G\setminus \{e\}$, which can be written as a sequence of at most~$l$ generators of~$G$, is mapped to the identity matrix.
The following theorem is adpated from~\cite{guth_lubotzky} and~\cite{siran}.
\begin{thm}\label{thm:}
	For any positive integer~$l$ there exists an $l$-locally faithful representation of $G$.
\end{thm}
\begin{proof}
	Let $\pi_I : \mathbb{Z}[\xi] \rightarrow \mathbb{Z}[\xi] /I$ be the quotient map for an ideal $I \subset \mathbb{Z}[\xi]$.
	Here we will only consider maximal ideals~$I$ so that~$\mathbb{Z}[\xi] /I$ is in fact a field $\mathbb{F}_q$ of characteristic~$p$.
	We note that $I$ is of the form $\langle p \rangle$ or $\langle p, g(\xi) \rangle$, where~$g$ is an irreducible factor of the minimal polynomial of $\xi$ in $\mathbb{F}_p$ (see \Cref{thm:ideal_types} in Appendix~\ref{sec:ideal_types}).
	
	Let $\rho$ be the representation of the infinite tessellation group defined by \Cref{eqn:rep_def}.
	We can extend $\pi_I$ to act on the coefficients of matrices over~$\mathbb{Z}[\xi]$.
	Since $G$ is generated by reflections it is easy to see that the matrices in $\text{im}(\rho)$ have determinant $\pm 1$.
	Since $1\notin I$ it follows that $\pi_I(\text{im}(\rho))$ only contains invertible matrices and hence we have that the function $\pi_I\circ \rho : G \rightarrow \text{GL}(D+1,\mathbb{F}_q)$ is well-defined.
	
	We will now show that for a suitable choice of the ideal $I$ the representation~$\pi_I\circ \rho$ is $l$-locally faithful.
	Let $g\in G\setminus \{e\}$ be a Coxeter group element which can be written as the product of $u\leq l$ generators, i.e. $g = a_{i_1}\dotsb a_{i_u}$.
	We need to show that~$\pi_I\circ \rho(g)$ is not the identity matrix.
	Clearly~$\rho(g)$ is not the identity matrix as $\rho$ is faithful.
	Furthermore, by choosing the prime $p$ in the ideal $I$ to be suitably large the image of~$\rho(g)$ under~$\pi_I$ is not the identity matrix.
\end{proof}

\subsubsection{Group structure}
It turns out that the group obtained by the procedure outlined above can have a particularly simple structure.
Assume that $I=\langle p \rangle$ and that $\pi_I(g)$ is non-singular.
As the elements of $\text{O}(1,5,\mathbb{Z}[\phi])$ preserve~$g$ we have that their images preserve~$\pi_I(g)$.
This means that $\text{im}(\pi_I)$ is a subgroup of~$\text{GO}_5(q)$, the orthogonal group over~$\mathbb{F}_q$.\footnote{The orthogonal groups over finite fields in odd dimensions are all isomorphic~\cite{wilson}.}
By diagonalizing~$g$ it is easy to see that $\pi_I$ is in fact surjective, so that we have $\rho \circ \pi_I \simeq \text{GO}_5(q)$.

The structure of $\text{GO}_5(q)$ for odd $q$ is well-known~\cite{wilson}: it decomposes into three simple groups as 
\begin{align}
	\text{GO}_5(q) \simeq \Omega_5(q)\rtimes (\mathbb{Z}_2\times \mathbb{Z}_2).
\end{align}
The group $\Omega_5(q)$ is also known as the Chevalley group $B_5(q)$ in the literature.
\begin{rem}
	Note that this decomposition is similar to the familiar one of the Lorentz group in $D=3$ into four connected components
	\begin{align}
	\text{O}(1,3,\mathbb{R}) \simeq \text{SO}^+(1,3,\mathbb{R}) \rtimes (\mathbb{Z}_2\times \mathbb{Z}_2)
	\end{align}
	where $\text{SO}^+(1,3,\mathbb{R})$ is the proper, orthochronous Lorentz group and $\mathbb{Z}_2\times \mathbb{Z}_2$ is generated by a time-like reflection (time reversal) and a space-like reflection.
\end{rem}

\subsubsection{Size of the quantum code}
The number of physical qubits~$n$ is given by the number of faces in the lattice.
We can count the number of faces (and cells of any other dimension) by counting the number of fundamental simplices and divide by the number of simplices per cell.
The number of fundamental simplices is the same as the order of the symmetry group of the lattice, which for odd~$q$ is given by the polynomial $|\Omega_5(q)\rtimes \mathbb{Z}_2| =  q^{10} - q^8 - q^6 + q^4$ \cite{wilson}.

For the $\{5,3,3,5\}$ tessellation there are 100 fundamental simplices per face and thus the formula for the size of a quantum code based on this construction is:
\begin{align}\label{eqn:code_size}
	n(q) = \frac{q^{10} - q^8 - q^6 + q^4}{100}
\end{align}
Note that we had to assume that $q$ is odd. We will later discuss examples with $q$ even for which \Cref{eqn:code_size} fails.

The golden ration $\phi$ has minimal polynomial $x^2-x-1$.
For~$p$ such that $x^2-x-1$ is irreducible in~$\mathbb{F}_p$ we  obtain $\mathbb{Z}[\phi] / \langle p \rangle \simeq \mathbb{F}_{p^2}$ and thus $n\in O(p^{20})$ in agreement with~\cite{guth_lubotzky}.

\subsection{Coverings}\label{sec:coverings}

In order to obtain more examples from the ones generated in previous sections we will now introduce topological coverings.
They will allow us to construct less symmetric examples as the group-based constructions.

\subsubsection{Definition}

In addition to the previous two methods for obtaining finite manifolds we employ a third method to construct small examples.
This method is based on \emph{coverings}: If $X$ and $C$ are topological spaces we say that~$C$ is \emph{covering space} if there exists a continuous surjective map $p:C\rightarrow X$ such that for any point $x\in X$ we have that there exists an open neighborhood~$U$ such that the pre-image~$p^{-1}(U)$ is a disjoint union of open sets in~$C$ each homeomorphic to~$U$.
If the number of these copies is fixed it is called the \emph{degree} of the covering.  We will call a covering of degree~$n$ an {\em $n$-fold covering}.
A famous example from physics is the 2-fold cover of~$\text{SO}(3)$ by~$\text{SU}(2)$.
In \Cref{fig:coverings} we show two further examples of coverings. The first (\Cref{fig:covering_circle}) is an infinite cover of the circle by the real line, depicted by putting the real line in a spiral over the circle so that~$p$ can be thought of as a projection along the vertical axis.
The covering an be realized by identifying $S^1$ with the unit circle in $\mathbb{C}$ and consider $p=\exp : \mathbb{R}\rightarrow S^1 \subset \mathbb{C},\, t \mapsto e^{it}$.
The second example (\Cref{fig:covering_torus}) is a 4-fold cover of a torus by another torus.
The covering is realized by taking translations in $x$- and $y$-direction modulo~20 and~10, respectively.

\begin{figure}
	\centering
	\subfloat[Infinite covering of circle~$S^1$ (blue) by real line~$\mathbb{R}$ (green).]{\includegraphics[width=0.45\linewidth]{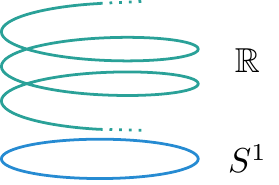}\label{fig:covering_circle}}
	\hfil
	\subfloat[4-fold covering of a $20\times 10$-torus (blue) by a $40\times 20$-torus (green).]{\includegraphics[width=0.45\linewidth]{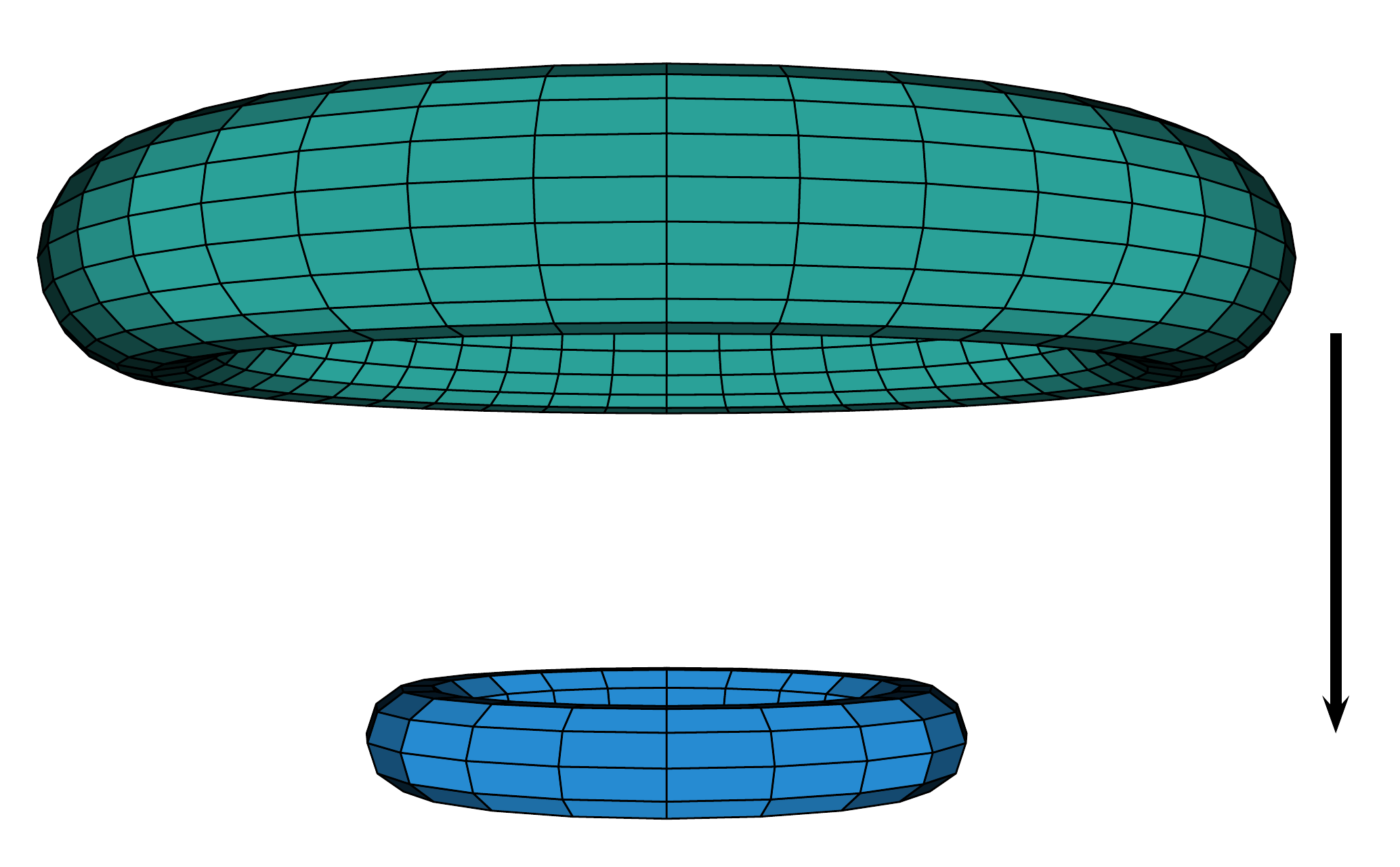}\label{fig:covering_torus}}
	\caption{(a)~The circle $S^1$ is covered by the real line.  The covering can be constructed by identifying $S^1$ with the set of complex numbers with unit 2-norm and defining $p=\exp : \mathbb{R}\rightarrow S^1 \subset \mathbb{C},\, t \mapsto e^{it}$. This is an example of an infinite covering, as the pre-image of any point has infinite cardinality. The deck transformation group is the abelian group $\mathbb{Z}$.
		     (b)~The small torus is covered by the larger torus.  The larger torus has 4 times the area of the smaller one. The covering is constructed by taking translations in the larger torus modulo the corresponding length in the smaller torus, in this case modulo~20 in the $x$-direction and modulo~10 in the $y$-direction. The preimage of any point contains four elements on which the deck transformation group $\mathbb{Z}_2 \times \mathbb{Z}_2$ operates.
	        }
	\label{fig:coverings}
\end{figure}

Coverings can be equipped with a group structure:
Homeomorphisms operating on the covering space $\phi: C \rightarrow C$ such that $p\circ \phi = p$ are called \emph{deck transformations}.
They form a group under composition called the \emph{deck transformation group}.
There is a natural group operation of the deck transformation group of a covering $p$ on the pre-image of a given point $x\in X$, as it permutes the elements of $p^{-1}(x)$.

What are the deck transformation groups in the two examples of \Cref{fig:coverings}? -- For the real line and the circle we can perform shifts by multiples of $2\pi$, i.e. $t\mapsto t+2\pi x$, leaving the image of~$t$ under~$\exp$ invariant.
The deck transformation group is hence isomorphic to the infinite abelian group~$\mathbb{Z}$.
For the covering of the torus it is clear that we can perform translations by~20 in the $x$-direction and translations by 10 in the $y$-direction leaving the modulus invariant.
Hence, the deck transformation group is isomorphic to~$\mathbb{Z}_2 \times \mathbb{Z}_2$.

\subsubsection{Finite coverings of hyperbolic 4-manifolds}
We construct further instances of hyperbolic 4-manifolds by enumerating all conjugacy classes of subgroups of the symmetry group of a given tessellated finite hyperbolic 4-manifold.
Not all subgroups preserve the local structure of the tessellation, as they may contain elements which have fixed points, such as reflections or rotations.
However, since we consider tessellations we have the more stringent restriction that the deck transformation group should  respect the local structure of the tessellation.
This restriction can be formulated in group-theoretic language and we call it the \emph{non-local subgroup condition} which we derive in Appendix~\ref{sec:quotient_condition}.

The deck transformation groups of the 4D hyperbolic manifolds constructed by coverings can be found in \Cref{tab:examples} under ``structure''.

\subsection{Examples of $\{5,3,3,5\}$-Codes}\label{sec:examples}

\begin{table*}[h]
	\centering
	\caption{Examples of $\{5,3,3,5\}$-codes.}\label{tab:examples}
	\bgroup
	\def\arraystretch{1.5}%
	\begin{tabular}{r r r c c r}
		\hline
		\hline
		\# & $n$ & $k$ & ideal & structure & Euler characteristic $\chi$ \\ 
		\hline
		1 & 144 & 72  & -- & $\left(\text{SL}_2(5)\rtimes \text{A}_5\right)\rtimes \mathbb{Z}_2$ & 26  \\ 
		2 & 720 & 184  & -- & 125-fold covered by 10, $\mathbb{Z}_5 \times \mathbb{Z}_5 \times \mathbb{Z}_5$ & 130  \\
		3 & 3,264 & 744 &  -- & 3-fold covered by 6, $\mathbb{Z}_3$ & 636  \\
		4 & 3,600 & 736 &  -- & 25-fold covered by 10, $\mathbb{Z}_5 \times \mathbb{Z}_5$ & 650 \\
		5 & 4,896 & 1,124 &  -- & 2-fold covered by 6, $\mathbb{Z}_2$ & 968  \\
		6 & 9,792 & 2,200 &  $\langle 2 \rangle$ & $\Omega_5(4)$ & 1,904  \\ 
		7 & 18,000 & 3,624 & -- & 5-fold covered by 10, $\mathbb{Z}_5$ & 3,250  \\ 
		8 & 18,432 & 4,232 &  -- & $\mathbb{Z}_2^{\times 8}\rtimes \left[\left(\text{A}_5\rtimes \text{A}_5\right)\rtimes \mathbb{Z}_2\right]$ & 3,584  \\ 
		9 & 19,584 & 4,324 & -- & $\Omega_5(4)\rtimes \mathbb{Z}_2$ & 3,808  \\ 
		10 & 90,000 & 18,024 &  $\langle \sqrt{5} \rangle$ & $\left[\left(\mathbb{Z}_5^{\times 4}\rtimes \text{SL}_2(5)\right)\rtimes \text{A}_5\right]\rtimes \mathbb{Z}_2$ & 16,250  \\ 
		11 & 34,432,128 & ? &  $\langle 3 \rangle$ & $\Omega_5(9)\rtimes \mathbb{Z}_2$ & 6,216,912 \\
		12 & 257,213,088 & ? & $\langle 11 \rangle$ & $\Omega_5(11)\rtimes \mathbb{Z}_2$ & 46,441,252 \\
		13 & 61,140,357,792 & ? & $\langle 19 \rangle$ & $\Omega_5(19)\rtimes \mathbb{Z}_2$ & 11,039,231,268 \\
		\hline
		\hline
	\end{tabular}
	\egroup
\end{table*}

Using the constructions of \Cref{sec:constructionRWS,sec:constructionPSL,sec:coverings} we have found  examples which are small enough to perform Monte Carlo simulations.
Here we will discuss the properties of these examples in more detail.
A summary can be found in \Cref{tab:examples}, where the properties of the tessellated manifolds and the associated quantum codes are listed.
The column labeled ``structure'' contains either the structure description of the associated symmetry group or, if the example was constructed from a finite covering, the number of the covering manifold and the deck transformation group.

\subsubsection{Based on FP-groups}
The construction based on finitely presented groups of \Cref{sec:constructionRWS} gave us three examples: the smallest with 144 physical qubits, one with 18,432 physical qubits and one with 19,584 physical qubits.
All are obtained by factoring out a single translation.
They are (in the same order as above): 
\begin{itemize}
	\item {\small $ababacbdedcbabacedcbaedced$}
	\item {\small $bedcbabedcbabedcbabedcbabedcbabedcba$}
	\item {\small $baedcbedcbabacbdcedcbabcedcbabacbded$}
\end{itemize}
For readability we have written the generators as $a,\dotsc , e$ instead of $a_i$ for $i=0,\dotsc ,4$ (cf. \Cref{sec:constructionRWS}).

The first example is known as the {\em Davis manifold} and was first described in~\cite{davismanifold}.
It can be constructed from a single 120-cell (cf.~\Cref{fig:120cell}) by taking opposing dodecahedra at the boundary and identifying them.\footnote{The same procedure in 3D corresponds to identifying opposing faces of a dodecahedron. Note that as opposed to 4D, in 3D these faces do not allign and different rotations to make the faces match give rise to topologically different 3-manifolds: the Poincar\'e homology sphere, the Seifert-Weber space and the 3D real projective space~\cite{weber}.}
Note that doing so we do not obtain a proper $\{5,3,3,5\}$ tessellation, as for example the number of 3 cells incident to the (unique) 4-cell is~60 instead of~120.
However, the incidence numbers involving 2-cells is the same as for a proper $\{5,3,3,5\}$ tessellation, so that stabilizer weights and qubit degrees are unaffected.
By construction, the number of 3-cells in the Davis manifold is~60.  The number of faces, edges and vertices is~144,~60 and~1 and thus its Euler characteristic (cf. \Cref{sec:guth_lubotzky_codes}) is $\chi = 26$.
We note that the Davis manifold gives rise to a trivial error detection code of encoding rate~$1/2$ and distance~2.

The two larger examples are not proper $\{5,3,3,5\}$ tessellations either as their 4-cells contain only 60 3-cells as well.
They have, as far as we are aware, not appeared in previous literature.

\subsubsection{Based on Linear Representations}
The construction based on matrix representations (see \Cref{sec:constructionPSL}) yielded several more examples.
Let us first consider the simplest example in which we reduce the matrices modulo 2, i.e. we factor out the ideal~$\langle 2 \rangle$.
Since $x^2-x-1$ is irreducible in~$\mathbb{F}_2$ we obtain a matrix group with coefficients in $\mathbb{F}_4$.
This gives rise to a quantum code with 9,792 physical qubits and 2,200 logical qubits.
We note that the underlying group is isomorphic to~$\Omega_5(4)$.
We can change the set of generators from reflections to rotations by taking products $a_i a_j$ as a new set of generators.
Factoring out~$\langle 2 \rangle$ from the group generated by rotations gives the group $\Omega_5(4) \rtimes \mathbb{Z}_2$ which defines a quantum code with 19,584 physical qubits and 4,324 logical qubits.
This is the same group that we found previously using finitely presented groups.

Next, we will consider an example where the minimal polynomial does become reducible:
Consider the ideal generated by~$\sqrt{5} = 2\phi -1 \in \mathbb{Z}[\phi]$.
The quotient~$\mathbb{Z}[\phi] / \langle 2\phi -1 \rangle$ turns out to be isomorphic to~$\mathbb{F}_5$.
The resulting quantum code has 90,000 physical qubits and 18,024 logical qubits.

The next largest examples are the ideals generated by~3,~11 and~19 (see~\Cref{tab:examples}).  However, these were too large to determine the number of encoded qubits.
We note that the encoding rate is close to the upper bound given in \Cref{eqn:encoding_rate}.

\subsubsection{Based on Finite Coverings}
Further small examples can be obtained by the covering procedure (see~\Cref{sec:coverings}).
Two of the coverings we found appeared in previous literature:
a 5-fold covering using the $n=90,000$ manifold (number~10 in \Cref{tab:examples}) had been found in~\cite{5335compact}.
It was also observed in~\cite{5335compact} that the $n=90,000$ manifold is a $625$-covering space of the Davis manifold.
Further examples are enumerated in~\Cref{tab:examples} where the covering and the deck transformation group are specified.

\subsubsection{Further remarks}
Since we do not have an expression for the number of encoded qubits~$k$ we constructed the boundary operators and obtained the dimension of the second homology group (see~\Cref{tab:examples}).
We observe that for the examples we constructed the encoding rate~$k/n$ is close to the lower bound of \Cref{eqn:encoding_rate}.
Unlike for 2D homological codes for which one can efficiently determine the distance~\cite{semihyperbolic} we are not aware of any efficient procedure to obtain the distance of higher-dimensional homological codes.
A randomized searching procedure yielded logical operators of weight~2 for the $n=144$ code, a logical operator of weight~12 for the $n=720$-code and a logical operator of weight~6 for the $n=3,264$ code.
Note that these are upper bounds as logical operators of smaller weight may still exist.

\section{Decoding and Performance}\label{sec:performance}

The stabilizer checks of a $\{5,3,3,5\}$ code correspond to dodecahedra in the primal tessellation ($Z$-checks) and in the dual tessellation ($X$-checks).
Each check acts on all of its adjacent qubits which correspond to the pentagonal faces.
The stabilizer checks of the code to fullfill non-trivial linear dependencies:
The boundary of a 120-cell contains dodecahedra and as the boundary itself is boundaryless it follows that the product of all checks belonging to a 120-cell has to vanish.
This can also be understood when we interpret the poset diagram in \Cref{fig:poset} as a Tanner graph.
The three levels in the middle form the quantum code, while levels 0 and 4 determine the linear dependencies of the checks.
Due to the linear dependencies the syndrome in a 4D code consists of closed loops.

Assuming that errors occur independently and homogeniously a good decoding strategy is minimum weight decoding.
Unfortunately, there is no known efficient algorithm which given a collection of loops in 4D lattice returns a minimum-weight surface which has these loops as its boundary.
Regardless, we can settle for a less optimal, but efficient solution.
It was observed in~\cite{DKLP} that a 4D code can be decoded by ``shrinking'' the syndrome loops.
In~\cite{hastings_decoder} it was shown that in 4D hyperbolic space a linear time decoding procedure exists.

In this paper we will consider two decoding stategies: the first decoder is based on a cellular automaton and the second on belief-propagation.

\subsection{Cellular Automaton Decoder}\label{sec:CAdecoder}

\subsubsection{Background}

Cellular automata can be used to implement a primitive decoding algorithm.
This type of decoder has several desirable features.
It can be implemented using very simple classical control, as it essentially performs a majority vote on a small number of input signals and sends a signal to perform a bit- or phase-flip based on the outcome.
This is extremely fast and dissipates little heat when compared to other decoding schemes, such as minimum-weight perfect matching.
This is important as it makes it possible to implement the classical control close to the qubits which, depending on the specific hardware implementation, have to be kept at temperatures of a few Kelvin down to hundreds of milli-Kelvin.

Using cellular automata to decode quantum codes was first suggested in~\cite{DKLP}.
Two different update rules have been used in the literature.
The first is a {\em majority-vote rule} which simply performs a bit-/phase-flip if more than half of the Z-/X-checks incident to the qubit are violated.
The second rule is called {\em Toom's rule}.
It only performs bit-/phase-flips if parity checks in a specified direction are violated.
Toom's rule was first introduced in the classical setting as non-equilibrium dynamics for the 2D Ising model, exhibiting the unusual property of a stable memory phase at non-zero temperature in the presence of a magnetic field \cite{grinstein}.
Toom's rule has been shown to perform better than the majority-vote rule~\cite{local_decoders} when applied to the 4D toric code with a hypercubic tessellation.
It has been generalized to other euclidean tessellations in~\cite{kubica_cellular}.
It is, however, not clear how to apply Toom's rule in hyperbolic space.
The reason for this is that to be well-defined it needs a distinguished direction and hence a notion of parallel lines which is consistent throughout the system.
In hyperbolic space a single line does not uniquely define a parallel line through any other point, as Euclid's fifth postulate famously does not hold in hyperbolic space.
We will therefore only consider the isotropic majority-vote rule.

\subsubsection{Monte Carlo}

We consider the independent bit-/phase-flip model, where each qubit is acted upon by Pauli-$X$ and Pauli-$Z$ each with probability~$p$.
We then run the CA decoder until the weight of the syndrome stops decreasing.
If the syndrome weight is non-zero we declare the trial a failure.
If the syndrome weight is zero we are back in a code state. In this case we check whether the error together with the recovery given by the CA decoder contains a non-trivial logical operator.

\begin{figure}[h]
	\centering
	\subfloat[linear]{\includegraphics[width=0.49\linewidth]{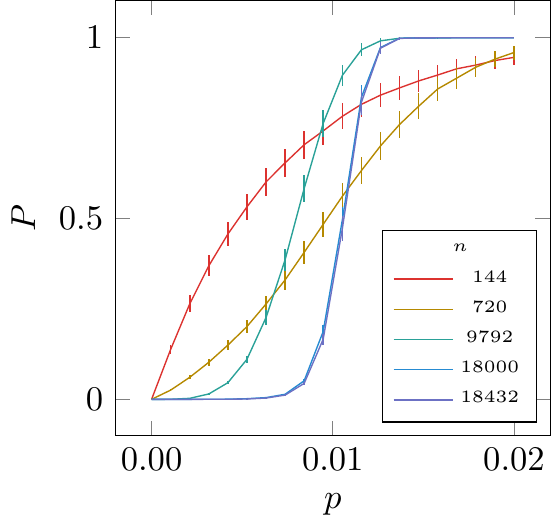}\label{fig:ca_linearplot}}
	\hfil
	\subfloat[log-log]{\includegraphics[width=0.49\linewidth]{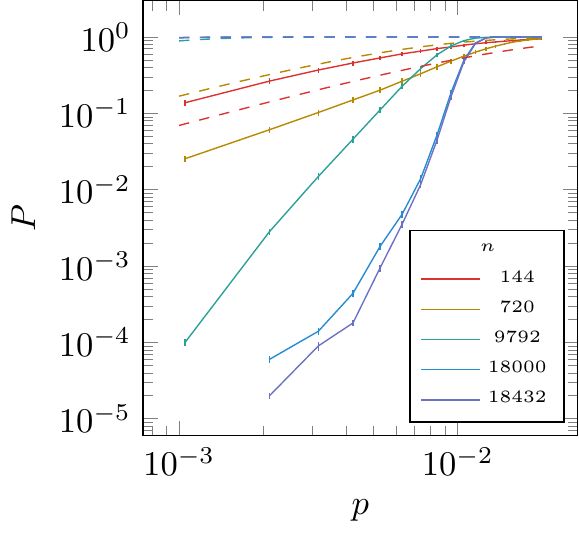}\label{fig:ca_logplot}}
	\caption{(a) Performance of CA decoder 
		(b)~Same data plotted with log-log axes. The dashed lines indicate the error probability if we were to take~$k$ unencoded qubits $1-(1-p)^k$.
	}
	\label{fig:ca_plot}
\end{figure}

The results of the simulation can be found in \Cref{fig:ca_plot}.
In~\cite{local_decoders} the same decoder under the same error model was applied to the 4D toric code, which is defined on the euclidean $\{4,3,3,4\}$ hypercubic tessellation.
We notice that the performance of the hyperbolic codes is better: for the 4D toric code the threshold error rate is below~$0.5\%$ while for the hyperbolic codes here errors are  suppressed for physical error rates below~$1\%$.
This is despite the fact that the hyperbolic code has higher stabilizer weight (12 instead of 6).

A quantity of interest is the {\em pseudothreshold}, which is the physical error probability below which the logical error probability is below the physical error probability, i.e. the error probability below which encoding is benificial over having bare qubits.
The error probability of~$k$ unencoded qubits is $1-(1-p)^k$ and marked in dashed lines in \Cref{fig:ca_logplot}.
We see that the $n=144$ code does not have a pseudothreshold which is expected, as it has distance 2 (cf.~\Cref{sec:examples}).
All other codes have a pseudothreshold of around~1\%.

\subsection{Belief-Propagation Decoder}\label{sec:BPdecoder}
The cellular automaton decoder of the previous section makes decisions based on a very limited amount of information: it can only see parity-check violations in its immediate vicinity.
This limitation is overcome by considering {\em belief-propagation} which is a commonly used decoding algorithm for classical LDPC codes~\cite{richardson_urbanke}.
Belief-rpropagation has been previously applied to quantum codes in~\cite{duclos2010fast} and~\cite{panteleev}.

\subsubsection{Background}
For Tanner graphs which are trees the Belief Propagation (BP) decoder corresponds to maximum likelihood decoding. 
As the Tanner graph of a $\{5,3,3,5\}$ quantum code is not a tree, BP gives a heuristic decoding algorithm in this setting.

In order to define BP, let $X^{(j)}$ be random variables corresponding to qubits.
With the assumed noise model, they are independently and identically distributed like Bernoulli variables with parameter $p \in [0,1]$. 
Let $Y^{(k)}$ be random variables corresponding to check nodes defined as
\begin{align}
	Y^{(k)} = \bigoplus_{\text{j neighbour of k}} X^{(j)}.
\end{align}
The values of $Y^{(k)}$ are what we observe when extracting the syndrome information and we denote them by $y^{(k)}_{\text{obs}}$.

We want to compute marginals of the random variables~$X^{(j)}$ conditioned on the observations $y^{(k)}_{\text{obs}}$.
If the Tanner graph were a tree, we could set one of the qubits to be the root of this tree. We will use the notation $k > j$ to denote that $k$ is a descendant of $j$.
For each qubit $j$, we define the following function whose domain is $\{0,1\}$:
\begin{align}
	p^{(j)} (x) = \Pr\left( X^{(j)} = x \mid \{ Y^{(k)} = y^{(k)}_{\text{obs}} \}_{k > j}\right)
\end{align}

For each check node $k$, denote by $j$ its parent qubit, we define the following function whose domain is $\{0,1\}$:
\begin{align}
	q^{(k)} (x) = \Pr\left(Y^{(k)} = y^{(k)}_{\text{obs}} \mid X^{(j)} = x \, , \, \{ Y^{(m)} = y^{(m)}_{\text{obs}} \}_{m > k}\right) 
\end{align}

To compute $p^{(j)} (x)$ from $(q^{(k)} (x))_{k \text{ children of } j}$, we need the following variation of Bayes' formula:
\begin{align}
	\Pr(A \, | \, B,C) \Pr(B \, | \, C) = \Pr(B \, | \, A,C) \Pr(A \, | \, C)
\end{align}

Indeed the left hand side of the above equation equals
\begin{align}
	\frac{\Pr(A \cap B \cap C)}{\Pr(B \cap C)}  \frac{\Pr(B \cap C)}{\Pr(C)} = \frac{\Pr(A \cap B \cap C)}{\Pr(C)}
\end{align}
which is symmetric in $(A,B)$ and therefore equals the right hand side. 
We apply this formula to the events
\begin{align}
\begin{split}
A &= \left( X^{(j)} = x \right) \\
B &= \left( \{ Y^{(k)} = y^{(k)}_{\text{obs}} \}_{k \text{ children of } j} \right) \\
C &= \left(  \{ Y^{(m)} = y^{(m)}_{\text{obs}} \}_{m > j \, , \, m \text{ not a child of } j} \right) 
\end{split}
\end{align}
and define the normalization constant $Z = \Pr(B \, | \, C)$. We know that $\Pr(A) = p$ and obtain
\begin{align}
\begin{split}
p^{(j)} (x) &= \Pr(A \, | \, B,C) \\
 &= \frac{p}{Z} \prod_{k \text{ children of } j} q^{(k)} (x). 
\end{split}  
\end{align}
Since $p^{(j)}(0) + p^{(j)}(1) = 1$, we obtain that
\begin{align}\label{eqn:qc_prob}
	Z = p \prod_{k \text{ children of } j} q^{(k)} (x) + (1-p) \prod_{k \text{ children of } j} q^{(k)} (1-x).
\end{align}

We now compute $q^{(k)} (x)$ from $(p^{(l)} (x))_{l \text{ children of } k}$.
We have
\begin{align}
q^{(k)} (x) &= \Pr\left(Y^{(k)} = y^{(k)}_{\text{obs}} \, | \, X^{(j)} = x \, , \, \{ Y^{(m)} = y^{(m)}_{\text{obs}} \}_{m > k}\right) 
\end{align}
giving
\begin{align}\label{eqn:cq_prob}
	1 - 2q^{(k)} (x) &= (-1)^{y^{(k)}_{\text{obs}} + x + 1} \prod_{l \text{ children of } k} (1 - 2 p^{(l)} (1) ). 
\end{align}

We could use \Cref{eqn:qc_prob,eqn:cq_prob} directly to define the iterative Belief Propagation algorithm. However for numerical stability reasons we will follow~\cite{richardson_urbanke} and use logarithmic ratios:

\begin{align}
\begin{split}
lp^{(j)} &= \log \frac{p^{(j)}(0)}{p^{(j)}(1)} \\
lq^{(k)} &= \log \frac{q^{(k)}(0)}{q^{(k)}(1)}
\end{split}
\end{align}

Under this transformation \Cref{eqn:qc_prob} translates into:
\begin{align} \label{qc llr}
lp^{(j)} = \log{\frac{1-p}{p}} + \sum_{k \text{ children of } j} lq^{(k)}.
\end{align}

Observing that $q^{(k)}(1) = (\exp{(lq^{(k)})}+1)^{-1}$, we obtain
\begin{align}
1 - 2q^{(k)}(1) &= \tanh{\frac{lq^{(k)}}{2}}.
\end{align}

Similarly $1 - 2p^{(j)}(1) = \tanh{(lp^{(j)} / 2)}$ and therefore \Cref{eqn:cq_prob} translates into:
\begin{align} \label{cq llr}
lq^{(k)} &= \frac{(-1)^{y^{(k)}_{\text{obs}}}}{2} \text{ argtanh } \left( \prod_{l \text{ children of } k} \tanh{\frac{lp^{(l)}}{2}} \right) 
\end{align}

The BP decoder we use is defined from \Cref{qc llr,cq llr}: the check node $k$ sends the message $lq^{(k)}$ to its parent node. The qubit $j$ sends the message $lp^{(j)}$ to its parent node. 
The first message is sent by the leaves of the tree, which we assume are qubits. It is initialized to $\log{(\frac{1-p}{p})}$.
The last message is received by the root of the tree, which we assume is a qubit. The value $(\exp{(lp^{(\text{root})})}+1)^{-1}$ gives the probability that the random variable corresponding to the root is~1 conditioned on the observation of all the check variables. 

The derivation assumed that the Tanner graph was a tree.
However, even for codes for which this is not the case we can still use the Belief Propagation algorithm as it was described above.
Although it does not compute exact probabilities any more: it is a heuristic whose performance we investigate numerically.

\subsubsection{Monte Carlo (perfect measurements)}

We first consider the setting where measurements can be performed perfectly, meaning without errors.
We apply the Belief Propagation decoder in parallel.
A round of message-passing consists in each qubit sending a message to each of its neighbor check node and each check node sending a message to each of its neighbor qubit.
After each round~$r$ of Belief Propagation we compute~$w_r$, the weight of the syndrome if we were to flip the qubits whose belief to have an error is higher than $0.5$. We stop as soon as $w_{r} \geq w_{r-1}$ or when $w_r = 0$. If we stopped because $w_r = 0$ and there is no logical error, we say that the decoding succeeded. \Cref{fig:bp_plot_n} shows the statistical frequency of unsuccessful decoding as a function of the physical error rate.

The data is consistent with a threshold above~5\% physical error rate.
However, we would like to note that due to the complicated dynamics of belief propagation it is generally hard to prove that a decoding threshold exists.
In fact it is known from classical coding theory that the performance of BP reaches an ``error floor'' for low physical error rates which occurs due to loops in the Tanner graph~\cite{richardson_urbanke}.
It is possible to eliminate this problem by postprocessing the output of BP with the {\em ordered statistics decoder}~(OSD) which has a computational complexity of~$O(n^3)$~\cite{panteleev}.
We have not done this here and leave it as future work.

\subsubsection{Monte Carlo (noisy measurements)}
To simulate noisy syndrome extraction we flip the syndrome with probability $q$.
For our simulations we chose $q=p$.
We consider $T$ rounds of error correction. 
In each round $t \in \{1, ..., T\}$, each qubit independently undergoes a $Z$ error $e_t^{noise}$ with probability $p$.
If $t \neq 1$, this error $e_t$ is added to $e_{t-1}^{res.}$, the residual error at round $t-1$.  The noiseless syndrome is computed:
$$s_t^{noiseless} = H(e_{t-1}^{res.} \oplus e_t^{noise}).$$
For $t \in \{1, ..., T-1\}$, each check node independently undergoes an error with probability $q$. This defines a syndrome noise $s_t^{noise}$.
The noisy syndrome is given to the BP decoder: 
\begin{align}
	s_t^{noisy} = s_t^{noiseless} \oplus s_t^{noise}.
\end{align}
The BP decoder outputs an inferred error:
\begin{align}
	e_t^{inf.} = \text{BP}_{\text{dec.}} (s_t^{noisy})
\end{align}
and the residual error is updated:
\begin{align}
	e_t^{res.} = e_{t-1}^{res.} \oplus e_t^{noise} \oplus e_t^{inf.}
\end{align}

For the last round, $t=T$, we assume perfect measurements and therefore have $s_T^{noise}=0$.
If the weight of the syndrome after the BP correction of this last round is zero and the residual error~$e_T^{res.}$ is not a logical error, we say that the decoding succeeded and otherwise, that it has failed.  
\Cref{fig:bp_plot_T} shows the statistical frequency of unsuccessful decoding against the physical error rate for $T = 5$. Note that the noiseless measurement scenario corresponds to $T=1$.
We see that increasing the system size decreases the logical error probability for a physical/syndrome error probability of $p=q$ up to about $4\%$.

In \Cref{fig:bp_plot_19584} we show the results of running the BP decoder on the same code ($n=19,584$) for different number of iterations~$T$.
The performance becomes worse as $T=1$ is essentially the noiseless case, however the recession of curves appears to recede with the number of time steps~$T$.

The threshold of the surface code under the same error model, but having to repeat the syndrome measurement for $d$ rounds and using a decoder with much less favourable computational complexity, is about $3\%$~\cite{wang}.
Clearly, this is not a fair comparison, as the ckeck weight of the $\{5,3,3,5\}$-code is three times higer than the one of the surface code.
However, one also needs to factor in the linear encoding rate and hence the reduction in overhead, i.e. the number of physical qubits that need to be spend to obtain a a given number of logical qubits and for a desired suppression of logical errors~\cite{semihyperbolic}.
We leave a detailed analysis as an open problem for future study.

\begin{figure}[h]
	\centering
	\subfloat[noiseless syndrome ($T=1$)]{\includegraphics[width=0.49\linewidth]{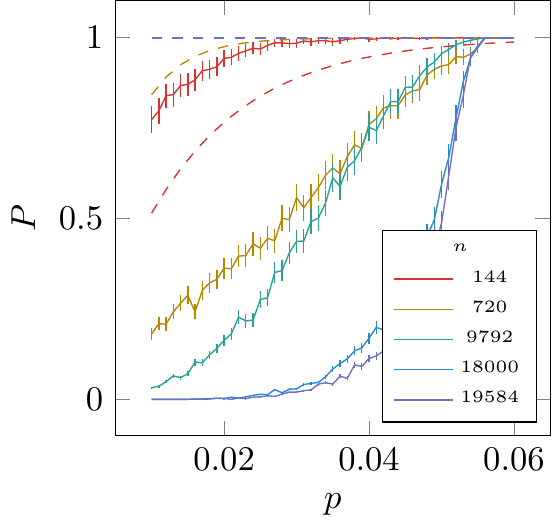}\label{fig:bp_plot_n}}
	\hfil
	\subfloat[noisy syndrome with $T=5$]{\includegraphics[width=0.49\linewidth]{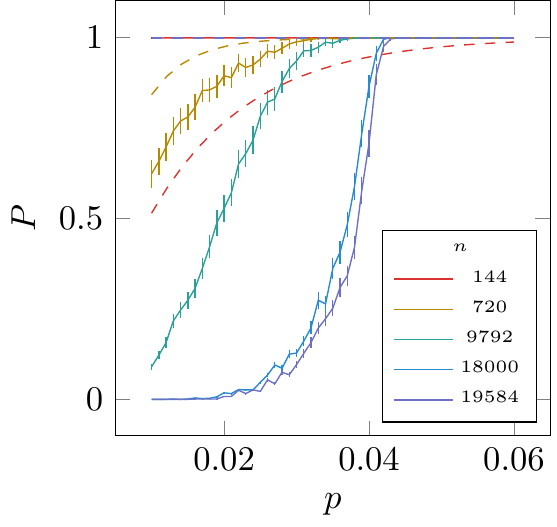}\label{fig:bp_plot_T}}
	
	\subfloat[$n=19,584$]{\includegraphics[width=0.49\linewidth]{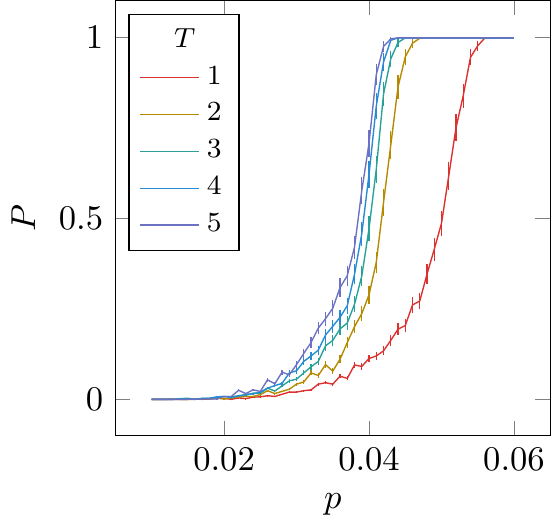}\label{fig:bp_plot_19584}}
	\caption{(a) Performance of BP decoder for a single time step $T=1$. We see that errors are suppressed in the system size for $p<0.5$. The curves cross very close to $P=1$. The dashed lines indicate the error probability if we were to take~$k$ unencoded qubits $1-(1-p)^k$. Vertical error bars correspond to the approximate 95 \% confidence interval given by $p = \hat{p} \pm 1.96 \sqrt{ \hat{p} (1-\hat{p}) / n_{\text{trials}} }$ where $\hat{p}$ is the mean. Here $n_{\text{trials}} = 1000$ for each physical error rate and each quantum code.
	(b)~Increasing the number of time steps to~$T=5$ lets the curves recede backwards. The pseudo-threshold for the two largest codes is around $4\%$. For the codes with more than 10 000 qubits, the transition between the successful and the unsuccessful decoding is quite sharp and gives numerical evidence for a noiseless threshold above 5\%.
	(c)~Performance curves for fixed system size $n=19,584$. The recession of the curves appears to slow down with the number of time steps~$T$.}
	\label{fig:bp_plot}
\end{figure}

\section{Conclusion}
We have shown how to construct quantum codes with a constant encoding rate and polynomial scaling distance from regular tessellations of four-dimensional hyperbolic manifolds.
Some of the manifolds we constructed were known, but many have (to the best of our knowledge) not appeared in previous literature.
We focussed on a particular tessellation of hyperbolic 4-space called the $\{5,3,3,5\}$ tessellation.
The resulting code family has an asymptotic encoding rate of $k/n \rightarrow 13/72 = 0.180...$, stabilizer checks of weight~12 and distance scaling polynomially as $\Theta(n^\epsilon)$.
For the construction based on linear representations it can be shown that $\epsilon \leq 0.3$, but we would like to stress that~$\epsilon$ may be higher when considering general $\{5,3,3,5\}$-codes.

Future work -- A drawback of our construction is the high stabilizer weight.
It was shown in~\cite{hastings_weight} that by refining the primal and dual tessellation of a homological code in a controlled way one can reduce the stabilizer weight while keeping the asymptotic code parameters invariant.
A related procedure is refining the hyperbolic tessellation using a euclidean tessellation as done for 2D hyperbolic codes in~\cite{semihyperbolic}.
It has to be determined how to do this in a systematic way in our construction.
Alternatively, it may be possible to define a subsystem version in which stabilizers can be measured by low-weight gauge operators.

We are confident that the decoding performance may be increased by considering more sophisticated decoding algorithms.
For example, it would be worthwile to combine the BP decoder with an ordered statistics decoder (OSD) to help in cases where BP alone would fail~\cite{panteleev}.
Another interesting avenue of research would be determining the threshold of the maximum-likelihood decoder by analyzing the associated 4D hyperbolic random-plaquette gauge model, as previously done for homological codes in euclidean space~\cite{DKLP,wang,chubb}.
However, it seems to be non-trivial to find a suitable order parameter, as bulk-boundary scaling works differently in hyperbolic geometry.

Finally, we hope to perform a detailed comparison to currently favoured quantum fault-tolerance schemes such as the ones based on the surface code.
The logical operators of hyperbolic codes share their support, making them much more efficient in terms of physical qubits for a fixed logical error rate~\cite{semihyperbolic}.
We blieve that the numerical single-shot performance, the simplified decoding and the high encoding rate make this LDPC code family highly competitive in situations where planarity is not an issue, such as in modular quantum computing architectures.

\appendices

\section{The Ideals of $\mathbb{Z}[\xi]$}\label{sec:ideal_types}

In \Cref{sec:constructionPSL} we discuss a construction based on a representation of the symmetry group of the $\{5,3,3,5\}$ tessellation.
The coefficients of the matrices in this representation contain linear combinations of powers of the golden ratio $\phi$.
The ring of these elements is denoted $\mathbb{Z}[\phi]$.
To obtain finite examples we map this representation to one over finite fields.
This is achieved by factoring out maximal ideals of $\mathbb{Z}[\phi]$.
An ideal is \emph{maximal} if it is a proper subset of the ring and all other ideals are contained in it.
It is well-known that the quotient of a ring with respect to a ideal is a field if and only if that ideal is maximal.
In what follows we discuss a slightly more setting where the ideals are prime.
An ideal $I$ is \emph{prime} if for any $a$ and $b$ with $ab \in I$ it holds that either $a\in I$ or $b\in I$.

We characterize the ideals of any ring of integers~$\mathbb{Z}[\xi]$ to which an element $\xi$ is added.
Let $h$ be the minimal polynomial of~$\xi$ in~$\mathbb{Z}[\xi]$.
As in \Cref{sec:matrix_entries} we will directly identify~$\mathbb{Z}[\xi]$ with~$\mathbb{Z}[x]/\langle h \rangle$.

\begin{lem}\label{lem:prime_ideal}
	Let $h\in \mathbb{Z}[x]$ be an irreducible polynomial and~$I$ a non-zero prime ideal of~$\mathbb{Z}[x]/\langle h \rangle$. Then~$I$ must contain a unique prime number~$p$.
\end{lem}
\begin{proof}
	Let $\tilde{I}$ be the preimage of $I$ under the natural epimorphism.
	Since $I\neq \langle 0 \rangle$ by assumption, $\tilde{I}$ must contain a $g\in \mathbb{Z}[x]$ which is not a multiple of $h$.
	Since $h$ is irreducible we must have $\gcd(h,g)=1$.
	By B\'ezout's identity there exist polynomials $u,v\in \mathbb{Z}[x]$ such that $uh+vg=p$, where~$p$ is a positive integer such that the $\gcd$ of~$p$ and the coefficients of~$u$ and~$v$ is~1.
	Since $I$ is a prime ideal, so is $I\cap\mathbb{Z}$ and since $p\in I$ we have that~$p$ must be a prime number.
\end{proof}

\begin{lem}\label{lem:finite_field}
	Let $h\in \mathbb{Z}[x]$ be an irreducible polynomial and~$I$ a non-zero prime ideal of~$\mathbb{Z}[\xi] = \mathbb{Z}[x]/\langle h \rangle$. Then the ring $\mathbb{Z}[\xi]/I$ is a finite field and~$I$ is a maximal ideal.
\end{lem}
\begin{proof}
	Since $I$ contains a prime number by \Cref{lem:prime_ideal} and since~$h$ (the minimal polynomial of $\xi$) has finite degree we have $|\mathbb{Z}[\xi]/I| < \infty$.
	Furthermore, since $I$ is a prime ideal we know that $\mathbb{Z}[\xi]/I$ has no zero divisors.
	
	We will now show that every non-zero element $a\in \mathbb{Z}[\xi]/I$ has a multiplicative inverse.
	Since $\mathbb{Z}[\xi]/I$ is finite there must exist positive integers $m$ and $n$ with $m>n$ such that $a^m = a^n$, which is equivalent to $a^n (a^{m-n}-1)=0$.
	Since $\mathbb{Z}[\xi]/I$ has no zero divisors this means $a^{m-n}=1$ which implies $a\cdot a^{m-n-1}=1$. Thus $a^{m-n-1}$ is the multiplicative inverse of~$a$.
\end{proof}

If $p$ is the prime number from \Cref{lem:prime_ideal} then $\mathbb{Z}[\xi]/I$ is a finite field of characteristic $p$.

\begin{thm}\label{thm:ideal_types}
	Let $h\in \mathbb{Z}[x]$ be an irreducible polynomial and~$I$ a non-zero prime ideal of~$\mathbb{Z}[\xi] = \mathbb{Z}[x]/\langle h \rangle$.
	Any prime ideal~$I$ of $\mathbb{Z}[\xi]=\mathbb{Z}[x]/\langle h \rangle$ is equal to 
	\begin{enumerate}
		\item $\langle p \rangle$, if $h$ is irreducible in $\mathbb{F}_p[x]$ or
		\item $\langle p, g(\xi) \rangle$, if $h$ is not irreducible in $\mathbb{F}_p[x]$,
	\end{enumerate}
	where  $p$ is a  prime number and $g$ is an irreducible factor of~$h\mod p$.
\end{thm}
\begin{proof}
	Let $\pi_I : \mathbb{Z}[\xi] \rightarrow \mathbb{Z}[\xi]/I$ be the quotient map.
	From \Cref{lem:finite_field} we know that $\mathbb{Z}[\xi]/I$ is a finite field and we note that it is generated (as a ring) by~$\pi_I(\xi)$.
	By \Cref{lem:prime_ideal} we have that $I$ must contain a unique prime number $p$.
	Hence, $I$ must be generated by $p$ and the minimal polynomial~$g$ of~$\pi_I(\xi)$.
	As~$\xi$ is a root of~$h$ in~$\mathbb{Z}[\xi]$ by construction, we must have that~$g$ divides $h \mod p$.
	The two cases in the statement of the theorem follow, depending on whether $g \equiv h \mod p$.
\end{proof}

\section{Quotient Condition}\label{sec:quotient_condition}

We use the notation of~\cite{abstractregularpolytpes}: let $G$ be a symmetry group generated by $(r_i)_{i \in \{0, …, 4\}}$. For $i \in \{0, …, 4\}$, let $S_i$ be the subgroup of $G$ generated by $(r_j)_{j \in \{0, …, 4\} \backslash \{i\} }$. Let $\mathcal{P}$ be the polytope associated with $G$ and $(S_i)_{i \in \{0, …, 4\}}$.  Let $H$ be a subgroup of $G$. We are interested in a condition sufficient to define a quantum code associated with the quotient polytope~$\mathcal{P}/H$. 

The orbit of an i-face $F_a = g_a S_i$ under the action of H is $\{ h g_a S_i \, | \, h \in H \}$. By definition this orbit is a face of the quotient abstract polytope. We denote it by $H F_a$. In terms of elements of $\Gamma$, It corresponds to the double coset $H g_a S_i$.\\
We use the same incidence definition: $H F_a$ and $H F_b$ are incident if $H g_a S_i \cap H g_b S_j \neq \varnothing$. \\

We want to find a condition under which quotienting on the right by $S_i \, , \, i \in \{0, ..., 4\}$ ``does not interact'' with quotienting on the left by $H$. More formally the following so-called non-local subgroup condition is sufficient to prove the lifting \Cref{lem:lifting_lemma}.

\begin{defi}[The non-local subgroup condition]
	We say that a subgroup of the symmetry group of the tessellation $G$ fullfills the \emph{non-local subgroup condition} if for any $i, j \in \{0, ..., 4\}$ and for all $g \in G$ we have
	\begin{equation}\label{eqn:NLC}
		gHg^{-1} \cap S_i S_j  = \{ \text{id} \}.
	\end{equation}
\end{defi}

The term non-local refers to the subgroup H: since the subgroups $S_i \, , \, i \in \{0, ..., 4\}$ are ``local'' (with respect for instance to the distance in the Caley graph $\left(\Gamma \, , \, (r_i)_{i \in \{0, ..., 4\}}\right)$, the subgroup H has to be non-local in order to not interact with the $S_i \, , \, i \in \{0, ..., 4\}$.

\begin{lem}[Lifting of $S_i$ cosets] \label{lem:lifting_lemma}
	Let H be a subgroup of a string C-group G of rank n satisfying the non-local subgroup condition (\Cref{eqn:NLC}). For $i, j \in \{0, ..., n\}$, let $H F_i$ be an $i$-face of $H \backslash \mathcal{P}_G$ and let $H F_j$ be a $j$-face of $H \backslash \mathcal{P}_G$ incident to $H F_j$.  \\
	For any $i$-face $K_i$ of $\mathcal{P}$ such that $H K_i = H F_i$, there exists a unique face $K_j$ of $\mathcal{P}$ such that $H K_j = H F_j$ and $K_j$ is incident to $K_i$.
\end{lem}

\begin{proof}
	There exists $g_i \in G$ such that $K_i = g_i S_i$. Then $H K_i = H F_i = H g_i S_i$. There exists $g_j \in G$ such that $H F_j = H g_j S_j$. There exist $h_i, h_j \in H$, $s_i \in S_i$ and $s_j \in S_j$ such that $h_i g_i s_i = h_j g_j s_j$. Therefore
	
	\begin{equation} \label{eqn:proof_a}
	g_i s_i = h_i^{-1} h_j g_j s_j
	\end{equation}
	
	We can define $K_j = h_i^{-1} h_j g_j S_j$. Clearly $H K_j = H F_j$ and $K_j$ is incident to $K_i$. \\
	
	To prove uniqueness suppose that a face $L_j$ of $\mathcal{P}$ satisfies $H L_j = H F_j$ and $L_j$ is incident to $K_i$. There exists $g \in G$ such that $L_j = g S_j$. 
	Since $H g S_j = H g_j S_j$, there exists $h \in H$ and $s''_j \in S_j$ such that
	
	\begin{equation} \label{eqn:proof_b}
	g = h g_j s''_j
	\end{equation}
	
	Since $L_j$ is incident to $K_i$, there exists $s'_i \in S_i$ and $s'_j \in S_j$ such that $g s'_j = g_i s'_i$. 
	Using \Cref{eqn:proof_a,eqn:proof_b}, we obtain $h g_j s''_j s'_j = h_i^{-1} h_j g_j s_j s_i^{-1} s'_i$. 
	We can rewrite this as $s_j^{-1} g_j^{-1} h_j^{-1} h_i h g_j s_j = s_i^{-1} s'_i (s'_j)^{-1} (s''_j)^{-1} s_j$. 
	
	Defining $\bar{g} = g_j s_j, \quad \bar{h} = h_j^{-1} h_i h, \quad \bar{s_i} = s_i^{-1} s'_i \quad  \text{and} \quad \bar{s_j} = (s'_j)^{-1} (s''_j)^{-1} s_j$, we have
	
	\begin{equation}
	\notag
	\bar{g}^{-1} \bar{h} \bar{g} = \bar{s_i} \bar{s_j}.
	\end{equation}
	
	Using the non-local subgroup condition \Cref{eqn:NLC}, it implies that $\bar{g}^{-1} \bar{h} \bar{g} = \text{id}$ and therefore that $\bar{h} = \text{id}$. 
	We have proven that $h = h_i^{-1} h_j$, which means that $L_j = K_j$.
\end{proof}

Thus under the non-local subgroup condition \Cref{eqn:NLC} we can use the lifting  \Cref{lem:lifting_lemma} to prove that the orthogonality of the parity-check matrices is preserved by such quotients: 

Let $H F_{i-1}$ be an $(i-1)$-face and $H F_{i+1}$ be an $(i+1)$-face of the quotient polytope. Let $\{ \tilde{F}_{i_1}, ..., \tilde{F}_{i_n} \}$ be the collection of $i$-faces incident to both $H F_{i-1}$ and $H F_{i+1}$. Under the non-local subgroup condition \Cref{eqn:NLC}, the  \Cref{lem:lifting_lemma} shows that there exist faces of the covering polytope $K_{i-1}$ covering $H F_{i-1}$, $K_{i+1}$ covering $H F_{i+1}$ and $\forall k \in \{1, ..., n\}, K_{i_k}$ covering $H F_{i_k}$ such that $\{ K_{i_1}, ..., K_{i_n} \}$ is the collection of $i$-faces incident to both $K_{i-1}$ and $K_{i+1}$. \\
The preservation of the orthogonality of parity check matrices would follow immediately from this. 

\section*{Acknowledgment}
NPB is supported by his UCLQ Fellowship.
NPB would like to thank Jens Eberhardt for helpful discussions on the ideals of $\mathbb{Z}[\phi]$, Leonid Pryadko for suggesting the randomized search procedure to obtain low-weight logical operators, Friedrich Rober for discussions on the LINS package in \textsc{GAP} and Dima Pasechnik for pointing out reference \cite{abstractregularpolytpes}.
VL would like to thank Antoine Grospellier, Lucien Grouès  and Anthony Leverrier for useful discussions on Belief Propagation.

\ifCLASSOPTIONcaptionsoff
  \newpage
\fi

\bibliographystyle{IEEEtran}
\bibliography{IEEEabrv,hyperbolic4D.bib}

\begin{IEEEbiographynophoto}{Nikolas P. Breuckmann}
	Nikolas P. Breuckmann holds a UCLQ Research Fellowship at University College London. He is interested in quantum information and related fields. He obtained his PhD at RWTH Aachen University working with Prof. Barbara Terhal on quantum fault-tolerance and quantum complexity theory. He has worked in industry at PsiQuantum, a Bay Area based start-up building a silicon-photonics based quantum computer.
\end{IEEEbiographynophoto}

\begin{IEEEbiographynophoto}{Vivien Londe}
	Vivien Londe is a Phd student from Inria Paris and University of Bordeaux. He is interested in constructions of LDPC quantum error correcting codes and their decoding. He is also interested in quantum and classical optimization algorithms.
\end{IEEEbiographynophoto}

\end{document}